%% file: unrelated.tex
\documentclass[11pt,fleqn]{article}

\usepackage{url}

\newcommand{\etal}{et al.\ }

\input{envs}

\newcommand{\cF}{\mathcal{F}}
\newcommand{\Sj}{j^*}

\newcommand{\Ss}{s^*}

\newcommand{\poi}{\textsf{Pois}}

\def\DEBUG{true}
\ifdefined\DEBUG
	
	\def\rem#1{{\marginpar{\raggedright\scriptsize #1}}}
	\newcommand{\sjr}[1]{\rem{\small\textcolor{red}{$\bullet${\tiny #1}}}}
	
	\newcommand{\remove}[1]{{\color{lightgray} #1}}
\else
	
	\newcommand{\sjr}[1]{}
	\new
	command{\remove}[1]{}
\fi
\def\DEBUG{true}
\ifdefined\DEBUG
	
	\def\rem#1{{\marginpar{\raggedright\scriptsize #1}}}
	\newcommand{\msr}[1]{\rem{\small\textcolor{red}{$\bullet${\tiny #1}}}}

\else
	
	\newcommand{\msr}[1]{}
	\new
	command{\remove}[1]{}
\fi

\newcommand{\cM}{\mathcal{M}}
\newcommand{\cG}{\mathcal{G}}
\newcommand{\cJ}{\mathcal{J}}

\begin{document}

\sloppy

\title{Weighted Completion Time Minimization for Unrelated Machines via Iterative Fair  Contention Resolution\thanks{This work 
is supported in part by NSF grants CCF-1409130, CCF-1617653, and CCF-1844939.}
}

\author{Sungjin Im\thanks{ Electrical Engineering and Computer Science, University of California, 5200 N. Lake Road, Merced CA 95344. \texttt{sim3@ucmerced.edu}.}  \and Maryam Shadloo\thanks{ Electrical Engineering and Computer Science, University of California, 5200 N. Lake Road, Merced CA 95344.  \texttt{mshadloo@ucmerced.edu}.}} 

\date{}
\maketitle
\thispagestyle{empty}

\begin{abstract}
	We give a 1.488-approximation for the classic scheduling problem of minimizing total weighted completion time on unrelated machines. This is a considerable improvement on the recent breakthrough of $(1.5 - 10^{-7})$-approximation (STOC 2016, 
	Bansal-Srinivasan-Svensson) and the follow-up result of $(1.5 - 1/6000)$-approximation (FOCS 2017, Li). Bansal \etal introduced a novel rounding scheme yielding strong negative correlations for the first time and applied it to the scheduling problem to obtain their breakthrough, which resolved the open problem if one can beat out the long-standing $1.5$-approximation barrier based on independent rounding. 
	Our key technical contribution is in achieving significantly stronger  negative correlations via \emph{iterative fair contention resolution}, which is of independent interest.  Previously,  Bansal \etal 	obtained strong negative correlations via a variant of  pipage type rounding  and Li used it as a black box. 
\end{abstract}
\clearpage
\setcounter{page}{1}
\input{1.intro}
\input{2.strong-negative}
\input{3.sched}
\input{4.analysis-sched0}
\input{4.analysis-sched1}

\input{4.analysis-sched}
\bibliographystyle{abbrv}
\bibliography{unrelated}
\end{document}

%% file: envs.tex
\usepackage{amsthm, amssymb}
\usepackage{graphicx}
\usepackage{algorithm}
\usepackage{algorithmic}
\usepackage{xspace}
\usepackage{enumitem}
\usepackage{tabularx}
\usepackage{xcolor}
\usepackage{graphicx} 
\usepackage{tikz}
\usepackage{wrapfig}
\usepackage[fleqn]{amsmath}

\DeclareMathOperator{\E}{\mathbb E}

\newtheorem{theorem}{Theorem}[section]
\newtheorem{lemma}[theorem]{Lemma}
\newtheorem{fact}[theorem]{Fact}
\newtheorem{observation}[theorem]{Observation}
\newtheorem{proposition}[theorem]{Proposition}
\newtheorem{corollary}[theorem]{Corollary}
\newtheorem{definition}[theorem]{Definition}

\newcommand{\eps}{\epsilon}

%% file: 1.intro.tex
\section{Introduction}

The unrelated machines setting is a classic scheduling model that has been widely used to model fully heterogeneous parallel machines. In this setting, there is a set $\cM$ of $m$ machines and  a set $\cJ$ of $n$ jobs to be scheduled on the machines. Machines are unrelated in the sense that each job $j \in \cJ$ has an arbitrary  size/processing time $p_{ij}$ on each machine $i \in \cM$. Further, each job $j$ has  weight $w_{j}$.\footnote{We can handle more general weights $w_{ij}$, which depend on machines, but we assume each job's weight is the same on all machines following the convention.} In this paper we consider non-preemptive---and therefore non-migratory---scheduling, which means that each job $j$ must be executed without interruption on one of the machines. In this paper we seek to optimize one of the most popular objectives, namely minimizing total weighted completion time, i.e., $\sum_{j \in \cJ} w_j C_j$ where $C_j$ denotes job $j$'s completion time. 

The problem we consider, denoted as $R||\sum_j w_j C_j$ using the common three-field notation, is known to be strongly NP-hard and APX-hard \cite{hoogeveen2001non}. For this problem, more than fifteen years ago, Schulz and Skutella \cite{schulz2002scheduling} gave an $1.5 + \eps$-approximation based on a time indexed LP.  Later,  Skutella \cite{Skutella01} and Sethuraman and Squillante \cite{sethuraman1999optimal} gave 1.5-approximations based on novel convex programming. It had been a long-standing open problem whether there exists a better than 1.5-approximation for the problem \cite{chekuri2004approximation,schulz2002scheduling,kumar2008minimum,sviridenko2013approximating,schuurman1999polynomial} until it was recently answered by Bansal \etal \cite{BansalSS16} in the affirmative.

The breakthrough by Bansal \etal had two important technical ingredients. First they introduced a novel SDP (semi-definite programming) to capture the pairwise interaction between jobs. Intuitively, this is important as the weighted completion time objective is not linear in job sizes. This is because a job $j'$ can delay another job $j$ if $j'$ starts its execution before $j$ on the same machine. Another contribution was developing a novel rounding scheme. All the aforementioned previous works used an independent rounding that randomly assigns each job $j$ to machine $i$ with probability $x_{ij}$, which is obtained by solving linear or convex programming. Since the 1.5-approximation factor is the best one can hope for using independent rounding, they had to develop a new rounding scheme. Their rounding not only ensures negative correlation between any pair of jobs assigned to the same machine but also ensures that any pair of jobs grouped together on the same machine subject to a capacity constraint are \emph{strongly negatively} correlated. Their rounding was based on a variant of pipage style rounding. Using a SDP relaxation and the new rounding with a delicate grouping of jobs, they were able to obtain a $1.5 - 10^{-7}$-approximation. 

Later, Li observed that a time-indexed LP can be used instead of a SDP \cite{Li17}.  A time-indexed LP (fractionally) encodes when each job starts and ends on each machine. Using the special structure of the time-indexed LP solution, he was able to use Bansal et al.'s strong negative correlation rounding with a different grouping and obtained a better $1.5 - 1/6000$-approximation. 

Up to date, the only way to obtain a better than 1.5-approximation for $R||\sum_j w_j C_j$ has been based on using the novel dependent rounding scheme by Bansal et al., which introduced the notion of strong negative correlations for the first time.

\subsection{Our Results}

\begin{theorem}
	\label{thm:completion-time}
	For minimizing total weighted completion time on unrelated machines, $R||\sum_j w_j C_j$, there exists a randomized 1.488-approximation.
\end{theorem}

As mentioned, this is a considerable improvement over the previous approximation ratios, $1.5 - 10^{-7}$ and $1.5 - 1/6000$ if we measure the improvement by the margin over the 1.5-approximations \cite{Skutella01,sethuraman1999optimal} that are based on independent rounding. The  improvement primarily comes from our new randomized dependent rounding that achieves strong negative correlations. We formally state what our randomized rounding guarantees as follows. We defer to Section~\ref{sec:application} an overview of how we apply  this theorem to $R||\sum_j w_j C_j$, which is similar at a high-level to  how Li \cite{Li17} applied Bansal et al.'s strong negative correlation rounding to the same scheduling problem.

\begin{theorem}
	\label{thm:rounding-rephrase}
Suppose we are given a set $\cM$ of machines and a set $\cJ$ of jobs together with a fractional assignment $\{x_{ij}\}_{i \in \cM, j \in \cJ}$ of jobs to machines (meaning that job $j$ is assigned to machine $i$ by $x_{ij}$) such that $x_{ij} \in [0, 1]$ for all $i \in \cM, j \in \cJ$ and $\sum_{i \in \cM} x_{ij} = 1$ for all $j \in \cJ$. For each machine $i \in \cM$, select any family $\cG_i$ of disjoint subsets of jobs such that $\sum_{j \in G} x_{ij} \leq 1$ for all $G \in \cG_i$. Then, there is a randomized rounding algorithm that achieves the following properties:
\begin{enumerate}
\item (Feasible Integer Assignment) Each job $j \in \cJ$ is assigned to exactly one machine $i \in \cM$, which is denoted as $i \leftarrow j$. 
\item (Preserving Marginal Probabilities) For every $i \in \cM$ and $j \in \cJ$, $\Pr[i \leftarrow j] = x_{ij}$.
\item (Negative Correlation) For every $i \in \cM$ and $j \neq j' \in \cJ$ such that no $G \in \cG_i$ has $j$ and $j'$ simultaneously, 
$\Pr[i \leftarrow j \wedge i \leftarrow j']  \leq x_{ij} x_{ij'}$
\item (Strong Negative Correlation\footnote{
As noted in \cite{BansalSS16}, it is impossible to impose strong negative correlation on every pair of jobs; thus, it was suggested to obtain strong negative correlations only between jobs in the same group. 
}) For every $i \in \cM$ and $j \neq j' \in \cJ$ such that  $j, j' \in G$ for some $G \in \cG_i$, 
$\Pr[i \leftarrow j \wedge i \leftarrow j']  \leq \frac{1}{1 + e} (e^{x_{ij}} + e^{x_{ij'}}) x_{ij} x_{ij'} \leq x_{ij}x_{ij'}$.
\end{enumerate}
The algorithm terminates in $O(mn \log n)$ time in expectation and w.h.p.
\end{theorem}

The theorem states the  properties our randomized rounding guarantees when rounding a fractional assignment of jobs to machines to an integer assignment.
The first property says the rounding always  outputs a feasible integer assignment when it terminates. 
The second property ensures that each job $j$ is assigned to machine $i$ with probability exactly $x_{ij}$, thus preserving the marginal probabilities. The third property says that two jobs $j \neq j'$ that are not grouped together on machine $i$ are assigned to the same machine $i$ with probability at most $x_{ij}x_{ij'}$; thus, the assignments are negatively correlated. All these properties can be achieved by simple independent rounding.

The last property, which is most interesting, says that for any two jobs grouped together on machine $i$, the probability they are simultaneously assigned to machine $i$ is at most $\frac{\exp( x_{ij}) + \exp( x_{ij'})}{e+ 1}x_{ij}x_{ij'} \leq x_{ij} x_{ij'}$ as $x_{ij} + x_{ij'} \leq 1$. So, if $x_{ij} + x_{ij'} $ is close to 1, the theorem doesn't guarantee strong negative correlation. However, it seems that we need strong negative correlation only when $x_{ij}$ and $x_{ij'}$ are sufficiently small, as will be discussed shortly. In the best scenario,  we have $\lim_{x_{ij}, x_{ij'} \rightarrow 0} \frac{\exp( x_{ij}) + \exp( x_{ij'})}{e+ 1} = \frac{2}{e+1} < 0.5379$.

We note that Theorem~\ref{thm:rounding-rephrase} is identical to Theorem 1.2 in \cite{BansalSS16} except that they have 
$\frac{107}{108} x_{ij} x_{ij'}$ instead of $\frac{1}{1 + e} (e^{x_{ij}} + e^{x_{ij'}}) x_{ij} x_{ij'}$ as the upper bound of the fourth property.  So, they have strong negative correlations regardless of value of $x_{ij} + x_{ij'}$. However, as alluded  above, Bansal \etal grouped $j$ and $j'$ together on machine $i$ only when $x_{ij}, x_{ij'} \leq 1/10$ and Li \cite{Li17} did only when $x_{ij}, x_{ij'} \leq 1/100$. Under their grouping of jobs, the coefficient of $x_{ij} x_{ij'}$ in the fourth property of Theorem~\ref{thm:rounding-rephrase} becomes $\frac{2 \exp(0.1)}{e+1} \simeq 0.5945$ and  $\frac{2 \exp(0.01)}{e+1} \simeq 0.5433$, respectively, meaning much stronger negative correlations when the theorem is actually needed.

Our method to achieve strong negative correlations is completely different from the method of Bansal \etal They used a very clever variant of pipage rounding. In contrast, we use a novel \emph{iterative fair contention resolution}. While our method is inspired by the fair contention resolution scheme by Feige and Vondr{\'{a}}k \cite{FeigeV06}, our method is considerably different and perhaps easier to understand, which we believe enables the iterative application of fair contention resolution. We discuss this in more detail in the following subsection.
\subsection{Our Techniques}

We first discuss two previous works that are most closely related to our new rounding (Theorem~\ref{thm:rounding-rephrase}).

\smallskip
\noindent
\textbf{A Variant of Pipage Rounding \cite{BansalSS16}.} As mentioned before, Bansal \etal \cite{BansalSS16} achieved the first strong negative correlations using a variant of pipage rounding. Note that the fractional assignment $\{x_{ij}\}_{i \in \cM, j \in \cJ}$ described in Theorem~\ref{thm:rounding-rephrase} can be thought of as a fractional matching of a bipartite graph over $(\cM, \cJ)$ saturating all jobs in $\cJ$. While there are many variations and extensions \cite{ageev1999approximation,arora1996new,chekuri2011multi,gandhi2006dependent,kahn1997stochastic}, a typical pipage rounding works as follows: In every iteration, the rounding finds a path or cycle only consisting of edges $(i,j)$ with fractional values, i.e., $x_{ij} \in (0,1)$. Then, it either maximally increases the weight of all odd (even, resp.) edges and simultaneously decreases the weight of all even (odd, resp.) edges by the same maximum margin---this choice is made at random in a way to preserve the marginal probabilities. This update  makes at least one more edge have value either 0 or 1; and the value of such edges remains fixed afterwards. Intuitively, two edges $(i,j)$ and $(i, j')$ incident to the same machine $i$ are chosen with negative correlation because the rounding never increases $x_{ij}$ and $x_{ij'}$ simultaneously. To obtain strong negative correlations additionally, in each iteration Bansal \etal carefully chose paths of length 4 based on a random 2-coloring of the edges.

\smallskip
\noindent
\textbf{Fair Contention Resolution \cite{FeigeV06}.} While our theorem statement is very similar to Bansal et al.'s corresponding theorem in \cite{BansalSS16}, our rounding scheme is completely different and is inspired by the fair contention resolution by Feige and Vondr{\'{a}}k \cite{FeigeV06}. To obtain a better than $1 - 1/e$-approximation for the Maximum Submodular Welfare problem\footnote{The goal of the Maximum Submodular Problem is to allocate items to players so to maximize the total utility where each player has a monotone submodular utility function.} and related problems, they developed a novel fair contention resolution scheme:  Suppose each player $a$ claims item $b$ independently with probability $q_{a,b}$.  Now we need to assign each item that was claimed by one or more players to exactly one of them. They gave an elegant contention resolution scheme where each  player $a$ gets an item $b$ with probability equal to $\frac{1 - \prod_{a'} (1- q_{a',b})}{\sum_{a'} q_{a',b}}$, conditioned on $a$ having claimed $b$. 
In the Maximum Submodular Welfare problem, if each player claims a subest of items according to her own distribution, the approximation guarantee was shown to be determined by the lowest probability that the player receives item $b$ conditioned on her having claimed it. The contention resolution scheme attempts to maximize the lowest probability across all players and items; thus, the name fair contention resolution comes.

As an attempt to obtain a theorem like Theorem~\ref{thm:rounding-rephrase} using the fair contention resolution scheme, think of each group $G$ defined in Theorem~\ref{thm:rounding-rephrase} as a player. 
Then, we let each group $G \in \cG_i$ claim a job $j \in G$ with probability $x_{ij}$ independently, hoping that this will help two different jobs in $G$ less likely be assigned to the same machine $i$ with the aid of the fair contention resolution. The resolution scheme  guarantees that $j$ is assigned to machine $i$ with probability at least 
$\frac{
1 - \prod_{i'} (1- x_{i'j}) 
}{
\sum_{i'} x_{i'j}
} 
x_{ij}
\geq (1 - \prod_{i'} e^{- x_{i'j}}) x_{ij}  =  (1 - 1/e) x_{ij}$. 
This seems like a good sign as the contention resolution preserves each group's choice up to $1 - 1/e$ factor.

Unfortunately, there are several issues. First, their rounding satisfies none of the desired properties claimed in Theorem~\ref{thm:rounding-rephrase}. It is not difficult to modify the scheme to satisfy the first two properties. However, their rounding scheme has no guarantees on the third and fourth properties. At a high-level, their algorithm focuses on the best guarantee on the first moment (recall that their goal was to give a better than $1 - 1/e$-approximation for some assignment problems) and their analysis is remarkably accurate. However, because of the very reason their algorithm and analysis don't  seem to readily extend to satisfy negative and strong negative correlations. 

\smallskip
\noindent
\textbf{Our Approach: Iterative Continuous Fair Contention Resolution.} We develop a new contention resolution scheme of a continuous flavor which we believe is perhaps more intuitive.\footnote{However, our method doesn't give as strong guarantee on the first moment as \cite{FeigeV06}. That is, using our method, each job $j$ is assigned to machine $i$ with probability exactly $(1 - 1/e) x_{ij}$ (in the first iteration). In contrast, the probability can be strictly greater in the method of \cite{FeigeV06} when $\{x_{ij}\}_i$ are not all tiny.} Thus, while the analysis is non-trivial, we are able to analyze its iterative application, thereby obtaining Theorem~\ref{thm:rounding-rephrase}.

Here, we sketch how we develop our fair contention resolution along with the intuitions behind.  To gain some intuitions, let's first focus on each job $j$. Instead of assigning $j$ to a machine $i$ with probability $x_{ij}$ as in independent rounding, we would like to have machine $i$ claim job $j$, taking other jobs into account,  with probability $x_{ij}$ in order to have a better control on jobs assigned to each machine. Then, an obvious issue is that job $j$ may be claimed by multiple machines (or by no machine). Therefore, we need to resolve this contention for job $j$ among machines. Towards this end, we generate $N_{ij}$ (lottery) \emph{tickets} for job $j$ on machine $i$, where $N_{ij} \sim \poi(x_{ij})$, a Poisson distribution with mean value $x_{ij}$. Note that this is equivalent to generating one ticket for job $j$ on machine $i$ with probability $\eps$ independently, for each $\eps$ unit of $x_{ij}$. Thus, this way of generating tickets allows us to view the problem more continutously. 
If no tickets are generated for $j$ across machines, which happens with probability $1/e$, job $j$ is not assigned; otherwise, we choose one ticket of job $j$ uniformly at random and assign the job to the machine from which the ticket originated. It is an easy exercise to see that $j$ is assigned to machine $i$ with probability exactly $(1 - 1/e)x_{ij}$.

Now, we want to impose strong negative correlations between jobs grouped together on each machine. Towards this end, we let each group $G \in \cG_i$ \emph{recommend} one job $j$ (or none) with probability $x_{ij}$ (if $j \in G$); here we use the fact that $\sum_{j \in G} x_{ij} \leq 1$. The idea is to ensure that $j$ has tickets on machine $i$ only when $G$ recommends job $j$ -- by doing so, if $j, j' \in G$, $j$ having tickets on $i$ will be negatively correlated with $j'$ having tickets on $i$. To continue to have the above nice contention resolution based on tickets,
our goal is to ensure $N_{ij} \sim  \poi(x_{ij})$ and $N_{ij} = 0$ unless $G$ recommends $j$ on machine $i$. To achive this, we use a simple trick. Let $\tilde N_{ij}$ denote the number of \emph{potential} tickets sampled from a distribution whose probability mass for each value $k >0$ is exactly $1/ x_{ij}$ times that of $\poi(x_{ij})$. Then, we set $N_{ij} := B_{ij} \cdot \tilde N_{ij}$ and have $N_{ij} \sim  \poi(x_{ij})$. Here,  $B_{ij}$ is an indicator random variable that has value one iff $j$ is recommended by the group where it belongs. In words, $j$ has $N_{ij} = \tilde N_{ij}$ \emph{real} tickets only when $B_{ij}=1$ occurs.

Using the above observation that each job is assigned to some machine with probability exactly $1 - 1/e$, we can repeat the whole process, excluding jobs that have already been assigned, until all jobs are assigned. 

The actual proof of the third and fourth properties is quite non-trivial, particularly the third property. This is because the random process of where two jobs $j$ and $j'$ are assigned depends on whether the two jobs are grouped together on each machine and how much they are assigned in the fractional solution. At a high-level, we show that the worst case for us happens when the two jobs $j$ and $j'$ are not grouped together on any machines possibly except on machine $i$---then the proof becomes relatively easy. To prove this we take a sequence of careful steps conditioning and deconditioning on some random variables. Perhaps proving negative correlations of our method is significantly more challenging than proving negative correlations of pipage rounding because our method assigns $(1 - 1/e)$-fraction of remaining jobs in each iteration unlike pipage rounding that assigns one job wlog in each iteration. Thus, we need to take a global view of the random process considering how each pair of jobs are grouped on all machines. 

\subsection{Applying Theorem~\ref{thm:rounding-rephrase} to the Scheduling Problem} 
	\label{sec:application}

To obtain Theorem~\ref{thm:completion-time} by applying Theorem~\ref{thm:rounding-rephrase} to our scheduling problem, we borrow some important ideas from Li's approach \cite{Li17}. Here we give a sketch of our rounding and briefly discuss the difference between our approach and his.  The time-indexed integer programming (IP) has an indicator variable $x_{ijs}$ which  is  1 if and only if $j$ starts its execution on machine $i$ at time $s$. Note that since we consider non-preemptive scheduling, if $x_{ijs} = 1$, then  $j$ completes at time $s + p_{ij}$.  Let $T$ be a sufficiently large upper bound on the number of time steps we need to consider. We assume $T$ is polynomially bounded in the input size since it was shown that this assumption is wlog with a loss of $(1+\eps)$ factor in the approximation ratio \cite{im2016better}. The time-indexed IP considered in \cite{Li17} is presented below.

Here, the first constraint ensures that every job is scheduled on some machine. The second constraint ensures that every machine processes at most one job at each time. The third constraint enforces that all jobs must complete by time $T$. By relaxing the last constraint  into $x_{ijs} \geq 0$, we obtain a valid LP relaxation. Let $x$ denote the optimal LP solution. Then, for each $x_{ijs} >0$, it will be convenient to think that we have a rectangle of height $x_{ijs}$ starting at time $s$ and ending at time $s + p_{ij}$, denoted as $R_{ijs}$. 
%
%
%\begin{figure*}[th!]
%\begin{wrapfigure}{r}{0.49\textwidth}
\begin{align}
\min \sum_{j \in \cJ} w_j &\sum_{i \in \cM, s \in [T]} x_{ijs} (s + p_{ij})\\
	\sum_{i \in \cM, s \in [T]} x_{ijs}  &= 1 &\forall j \in \cJ  \\
	\sum_{j \in \cJ, s \in (t - p_{ij}, t]} x_{ijs} &\leq 1 &\forall i \in \cM, t \in [T] \\
	x_{ijs} & = 0 &\forall i \in \cM,j \in \cJ,s > T - p_{ij} \\
	x_{ijs} & \in \{0, 1\} &\forall i \in \cM,j \in \cJ,s  \in [T]
\end{align}
%\end{figure*}
%\end{wrapfigure}

We first review how the independent rounding gives a 1.5-approximation: for each job $j$, we choose a rectangle $R_{ijs}$ independently with probability $x_{ijs}$ -- then, $j$ is assigned to machine $i$. Next, we sample a random offset $\tau_j$ for each job uniformly at random from $[0, p_{ij}]$. Then, we set $\theta_j = \tau_j +  s$ conditioned on $R_{ijs}$ being chosen for $j$. Now schedule jobs assigned to the same machine in increasing order of $\theta_j$. To upper bound the expected completion time of job $j$, we need to know the expected size of jobs that are assigned to $i$ and have smaller $\theta$ values than job $j$. Due to the linearlity of expecation, we can focus on the expected size of each job $j'$ that is scheduled prior to $j$ on machine $i$, which we call the expected delay $j'$ causes to job $j$ on machine $i$. Thanks to  the uniform choice of $\theta, \tau$ values, one can show that fixing $\theta_j$, the expected delay $j'$ causes to job $j$ on machine $i$ is exactly the area  of rectangle $R_{ij's'}$ up to time $\theta_j$, if $R_{ij's'}$ is chosen for $j'$. This is becuase the probability that $j'$ has a smaller $\theta$ value than $j$ on the condition is exatly $\frac{ \max ( \theta_j - s', 0)   }{p_{ij'}}$, which is exacty the horizontal length of $R_{ij's'}$ up to the time $\theta_j$, divided by $p_{ij'}$. 
Since the total area of rectangles till time $\theta_j$ is at most $1 \cdot \theta_j = \theta_j$ (recall at most one unit of job can be scheduled at each time), we  have $\E[ C_j \; | \; \theta_j, R_{ijs}] \leq \theta_j + p_{ij}$. Knowing that $\E[ \theta_j \; | \; R_{ijs}]  = s + p_{ij} /2$, we have 
$\E[ C_j \; | \;  R_{ijs}]  \leq s + 1.5 p_{ij}$. Further, de-conditioning on the choice of $j$'s rectangle, we have $\E[ C_j] \leq \sum_{i,j,s}  x_{ijs} (s + 1.5 p_{ij} )$, which immediately gives a 1.5-approximation due to the linarity of expectation.

Now we sketch how we use strong negative correlations to obtain a better than $1.5$-approximation. For the purpose of better intuitions, let's assume that there is  \emph{at most} one rectangle $R_{ijs}$ with $x_{ijs} > 0$ for every pair of job $j$ and machine $i$. In other words, assume that each job has at most one rectangle on machine $i$. Further, assume all jobs have weight 1.
Note that in the upper bound of $\E[ C_j \; | \;  R_{ijs}]$, the coefficient of $s$ is just 1. Therefore, if $s$ is  considerably big compared to $p_{ij}$, then job $j$ is an easy job on machine $i$ towards obtaining a better than 1.5-approximation. Another case of job $j$ being easy on machine $i$ is when $x_{ijs}$ is large. This is because in fact we actually have a better upper bound of $\E[ C_j \; | \; \theta_j, R_{ijs}] \leq \theta_j + p_{ij} - \tau_j x_{ijs}$. The reason why we get the extra negative term is as follows: Recall that the expected delay other jobs cause to job $j$ on machine $i$ is upper bounded by $\theta_j$, which is a clear upper bound on the total area of rectangles of other jobs up to time $\theta_j$. Here, we can take off 
the area of $R_{ijs}$ of job $j$ before time $\theta_j$, which is exactly $(\theta_j - s) x_{ijs} = \tau_j x_{ijs}$. 
 Roughly speaking, the rectangle $R_{ijs}$ of each bad job $j$ starts near time 0 and has small height $x_{ijs}$. This was one of the key observations made by Li \cite{Li17} although the definition of easy jobs is slightly different from ours. 

To obtain a better than 1.5-approximation, we need to use Theorem~\ref{thm:rounding-rephrase} (or the corresponding theorem in \cite{BansalSS16}). To handle bad jobs on machine $i$ using Theorem~\ref{thm:rounding-rephrase}, we need to group them carefully. At a high level, we  group jobs of similar $\theta$ values. That is, we first sample $\theta$ value for each rectangle of a bad job, and we group jobs if their $\theta$ values fall into the same time interval, which is one of the intervals of exponentially increasing length partitioning the whole time horizon. Using the fact that $x_{ijs}$ is small for bad jobs and the rectangle of bad jobs starts near time 0, with some care, we can bound the total $x$ value of jobs in each group, which is needed to apply Theorem~\ref{thm:rounding-rephrase}. Then, thanks to strong negative correlations, when two jobs $j$ and $j'$ have similar $\theta$ values, they are less likely to be assigned to the same machine $i$ with a good probability, which is enough to give a better than 1.5-approximation. 

While our rounding is different from Li's in many places, the two main differences are as follows. First, we use a random partition of time horizon into the intervals of exponentially increasing lengths whileas Li used a  deterministic partition. This is because we found the random partition seemed to give a better grouping of jobs. Further, for analysis, we conceptually group non-overlapping rectangles, so that we have a linear combination of subsets of non-overlapping rectangles, which looks like a solution to the configuration LP \cite{sviridenko2013approximating}. Using this structure also helps to improve the approximation ratio slightly. However, as mentioned before, the improvement of approximation ratio primarily comes from our stronger negative correlations.

\subsection{Other Related Work}

%%%%%%
Minimizing total (weighted) completion time is one of the most  popular scheduling objectives considered in the literature. 
For the single machine case, the algorithm highest-density-first, which favors jobs of highest $w_j / p_j$, is known to be optimal \cite{smith1956various}. The problem becomes NP-hard when there are multiple machines \cite{garey2002computers}. However, when machines are   identical ($P||\sum_{j}w_jC_j$) or uniformly related ($Q||\sum_{j}w_jC_j$), the problem admits PTASes \cite{afrati1999approximation,skutella2000ptas,chekuri2001ptas}. Interestingly, even when machines are unrelated, if the objective is to minimize total unweighted completion time ($R||\sum_{j} C_j$), the problem is polynomially solvable using a min-cost bipartite matching \cite{horn1973minimizing,bruno1974scheduling}. 
Kalaitzis \etal \cite{Ola2017unrelated} considered an important special case when each job's processing time is proportional to its weight, i.e., $R| w_{ij} / p_{ij} = 1 |\sum_{i,j} w_{ij} C_j$, and gave a 1.21-approximation. Interestingly, their result also achieves a bi-criteria 2-approximation for the makespan objective.
A configuration LP, which encodes all possible scheules on each machine, was 
shown to be solvable optimally within $(1+\eps)$-factor \cite{sviridenko2013approximating}; see \cite{im2016better} for a discussion on the strength of the configuration LP. 
The configuration LP was shown to have an integrality gap of at least 1.08 \cite{Ola2017unrelated}.
If jobs have different arrival times, the problem, $1|r_j|\sum_j w_j C_j$, is NP-hard \cite{lenstra1977complexity} even in the single machine case. When machines are identical ($P|r_j|\sum_j w_j C_j$) or related ($Q|r_j|\sum_j w_j C_j$),  the problem admits PTASes \cite{afrati1999approximation,chekuri2001ptas}. When machines are unrelated ($R|r_j|\sum_j w_j C_j$), 2-approximation \cite{schulz2002scheduling,Skutella01} had been the best approximation known for long until it was recently improved to 1.8687-approximation \cite{im2016better}. 

For the makespan objective on unrelated machines, i.e., $R||\max_j C_j$, a classic 2-approximation is known \cite{lenstra1990approximation}. There have been considerable efforts to improve this ratio for some special cases, e.g. \cite{svensson2012santa,ebenlendr2008graph}. For the dual objective of maximizing the minimum total load of all machines, see~\cite{bansal2006santa,asadpour2010approximation,asadpour2008santa,chakrabarty2009allocating,feige2008allocations}. For $R||(\sum_j (C_j)^k)^{1/k}$, see~\cite{azar2005convex,kumar2009unified}.  For $R|r_j|\sum_j (C_j - r_j)$, a poly-logarithmic approximation is known \cite{BansalK15}. For the special case of $R|r_j, p_{ij} \in \{p_j, \infty\}|\sum_j (C_j - r_j)$, see~\cite{GargK07,GargKM08}.  For a survey of other approximate scheduling results, see~\cite{chekuri2004approximation}.

\subsection{Organization}

We present our iterative fair contention resolution scheme in Section~\ref{sec:iter-algorithm}. Next, in Section~\ref{sec:iterative-analysis}, we show that the rounding scheme satisfies all the properties claimed in Theorem~\ref{thm:rounding-rephrase}.
We present our randomized rounding algorithm for the unrelated machines scheduling problem in Section~\ref{sec:unrelated-algo}.
We give the analysis of the algorithm in Section~\ref{sec:unrelated-analysis}, thereby proving Theorem~\ref{thm:completion-time}.

%% file: 2.strong-negative.tex
\section{Rounding Procedure Yielding Strong Negative Correlations}
	\label{sec:iter-algorithm}

\subsection{Preliminaries}

To describe our randomized rounding we need to introduce a probability distribution. Let $\poi(\lambda)$ denote the Poisson distribution with mean value $\lambda$. Recall that $\poi(\lambda)$ has pmf, $e^{-\lambda} \frac{\lambda^k}{k!}$, over $k \in \mathbb{Z}^+ := \{0, 1, 2, \ldots\}$. The following is a well-known property of Poisson distribution.

\begin{fact}
	\label{fact:poi}
	For any set of independent random variables $\{Z_h \sim \poi(\lambda_h)\}_{h \in [H]}$, it is the case that $\sum_{h \in [H]} Z_h \sim \poi(\sum_{h \in [H]} \lambda_h)$.
\end{fact}

For any $\lambda > 0$, let $\widetilde \poi(\lambda)$ denote the probability distribution with the following pmf: 

$$
\begin{cases} e^{-\lambda} \frac{\lambda^{k-1}}{ k!} &\mbox{if $k \in \{1, 2, 3, \ldots\}$} \\ 
1 - \frac{1 - e^{-\lambda} }{\lambda} & \mbox{otherwise, i.e., $k = 0$} \end{cases}
$$

This pmf is well-defined since $\sum_{k \geq 1} e^{-\lambda} \frac{\lambda^k}{\lambda k!} = \frac{1 - e^{-\lambda} }{\lambda} \leq 1$.\footnote{Here, we used the well-known facts that $e^x = \sum_{k \geq 0} \frac{x^k}{k!}$ and $e^x \geq x+1$.} Note that $\widetilde \poi(\lambda)$'s probability mass for each value $k > 0$ is $1 / \lambda$ times that of $\poi(\lambda)$. 

\begin{observation}
	\label{o:between-two-dist}
	Let $\tilde N \sim \widetilde \poi(\lambda)$ and $B$ be a Bernoulli (0-1) random variable with mean value $\lambda$. Then, $\tilde N \cdot B \sim \poi(\lambda)$.
\end{observation}
\subsection{Rounding Algorithm}
We are now ready to describe our randomized rounding. To simplify notation, for every $i \in \cM$ and $j \in \cJ$, if $j \notin G$ for all $G \in \cG_i$, then we create a singleton set of job $j$ and add it to $\cG_i$. Note that this has no effect on the properties we aim to prove as we do not change the existing groups. So, we can assume wlog that $\uplus_{G \in \cG_i} G = \cJ$ for all $i \in \cM$. 

Our randomized rounding is iterative. In the first iteration we perform as described in Figure~\ref{tabel1}. In principle, we do not need to generate potential tickets for job $j$ on machine $i$ unless $B_{ij} = 1$. However, the  description where we first generate potential tickets independently for every pair of job $j$ and machine $i$ makes the analysis more intuitive. For brevity,  we will interchangeably use $B_{ij}$ and $B_{ij} = 1$.

\begin{figure*}[!h]
\label{tabel1}
\vspace{-2ex}
\begin{center}
\fbox{
    \parbox{0.98\textwidth}{
\begin{enumerate}
\setlength{\itemsep}{0.3pt}
\setlength{\parskip}{0.3pt}
\item  For each $i \in \cM$ and $j \in \cJ$, let $\tilde N_{ij} \sim \widetilde \poi(x_{ij})$ be an independent RV; if $x_{ij} = 0$, then $\tilde N_{ij} = 0$. In words, we generate $\tilde N_{ij}$ \emph{potential} tickets for each job $j$ on machine $i$, according to $\widetilde \poi(x_{ij})$, independently.

\item For each $i \in \cM$ and $G \in \cG_i$, let $B_G \in G \cup \{nil\}$ be an independent RV such that $\Pr[ B_G = j] = x_{ij}$ for all $j \in G$. We use an indicator variable $B_{ij}$ to denote the event $B_G  = j$.  In words, each group $G \in \cG_i$ \emph{recommends} one job $j$ in the group $G$ with probability $x_{ij}$ independently and the event is denoted as $B_{ij}$; or it may recommend no jobs. 

\item For each $i \in \cM$ and $j \in \cJ$, let $N_{ij} = B_{ij} \cdot \tilde N_{ij}$. In words, all $\tilde N_{ij}$ potential tickets of job $j$ on machine $i$ become \emph{real} tickets if and only if the event $B_{ij} = 1$ occurs. Job $j$ has real tickets on machine $i$ iff $B_{ij} = 1$ and it has non-zero potential tickets on the machine. 

\item Each job $j \in \cJ$ is \emph{assigned} to machine $i \in \cM$, denoted as $i \leftarrow_1 j$, independently with probability $\frac{N_{ij}}{\sum_{i' \in \cM} N_{i'j}}$; if $\sum_{i' \in \cM} N_{i'j} = 0$, then $j$ is not assigned, denoted as $nil \leftarrow_1 j$. In words, among all real tickets of each job $j$ across machines, we choose one uniformly at random. If the ticket was generated on machine $i$, then we \emph{assign} job $j$ to machine $i$. 
\end{enumerate} 
    }
}
\caption{The first iteration of our randomized rounding}
\end{center}
\end{figure*}

We now explain how the rounding works in the subsequent iterations, $2, 3, \dots$. Every job assignment is \emph{final}. If job $j$ is assigned to machine $i$ in the $\ell$-th iteration (denoted as $i \leftarrow_\ell j$) the job is never considered in the subsequent iterations, $\ell+1, \ell+2, \dots$. Let $\cJ_{\leq \ell}$ denote the set of jobs that were assigned in iterations $1, 2, \dots, \ell$, i.e., $\cJ_{\leq \ell} := \{ j \in \cJ \; | \; i \leftarrow_{\ell'} j \textnormal{ for some } i \in \cM, \ell' \leq \ell\}$.
At the beginning of  the $\ell+1$-th iteration, we update $\cJ$ to $\cJ \setminus \cJ_{\leq \ell}$ and every set $G \in \cG_i$  to $G \setminus \cJ_{\leq \ell}$. Then, we perform the above four steps---all the RVs used in this iteration are different from those used in the previous iterations. The $\ell$-th iteration begins only when there exists a job that hasn't been assigned yet, i.e., $\cJ_{< \ell} \neq \cJ$. This completes the description of our randomized rounding.
\section{Proof of Theorem~\ref{thm:rounding-rephrase}}
    \label{sec:iterative-analysis}

This section is devoted to proving Theorem~\ref{thm:rounding-rephrase}. We first make an easy observation which will be useful in the analysis. Since we renew all RVs in each iteration, we add superscript $\ell$ to RVs if they are of the $\ell$-th iteration---for example, $N^\ell_{ij}$ denotes the RV in the $\ell$-th iteration corresponding to $N_{ij}$. For RVs of the first iteration, we omit the superpscript. 

\begin{observation}
	\label{obs:indep-rounds}
	Conditioned on  no jobs in $\cJ' \subseteq \cJ$ having been assigned in the previous iterations, the stochastic process of assigning $\cJ'$ is identical to the same process starting from the first iteration. 
\end{observation}

In particular, this observation means: For any event $E(\cJ', \ell)$
	concerning the assignment of some jobs $\cJ' \subseteq \cJ$ in the $\ell$-th iteration or in the subsequent iterations, we have $\Pr[ E(\cJ', \ell) \; | \; \cJ' \cap \cJ_{\leq \ell-1} = \emptyset] = \Pr[ E(\cJ', 1)]$. For example, consider the fourth property. Then, the observation implies that $\Pr[ i \leftarrow_{\geq \ell} j \wedge i \leftarrow_{\geq \ell} j' \; | \; j, j' \notin \cJ_{\ell-1}] = \Pr[ i \leftarrow j \wedge i \leftarrow j']$. Here $i \leftarrow_{\geq \ell} j$ denotes $j$ being assigned to $i$ in the $\ell$-th iteration or later. We illustrate Observation~\ref{obs:indep-rounds} using this as an example. Note that the assignment of $j$ and $j'$ are completely determined by  $\{\tilde N_{i'j}\}_{i' \in \cM}, \{\tilde N_{i'j'}\}_{i' \in \cM}$, and $B_G$ for all groups $G$ including $j$ or $j'$. Further, for $B_G$, what only matters is whether $B_G = j$ or $j'$, or not. These events, $\tilde N_{i'j}$, $\tilde N_{ij}$ are not affected by whether other jobs have been assigned or not. Extending this argument to Observation~\ref{obs:indep-rounds} is straightforward.

\smallskip
	We now shift to proving the four properties and the running time guarantee.

\subsection{First and Second Properties and Running Time}
	\label{sec:second}
	
	The first property says that the randomized rounding assigns each job to exactly one machine. In each iteration the rounding algorithm attempts to assign jobs that have not been assigned in the previous iterations. Therefore, the first property immediately follows if we show that the algorithm eventually terminates. We will show that the algorithm terminates in $O(\log n)$ iterations with high probability and also in expectation, after showing the second property.

For better readability we first give the analysis for the first iteration and extend it to arbitrary iterations using Observation~\ref{obs:indep-rounds}. We first make two easy observations. 

\begin{observation}
	\label{o:num-tickets}
	For any $i \in \cM$ and $j \in \cJ$, $N_{ij} \sim \poi(x_{ij})$.
\end{observation}
\begin{proof}
	The algorithm defines $N_{ij} = B_{ij} \cdot \tilde N_{ij}$, where $\tilde N_{ij} \sim \widetilde \poi(x_{ij})$ and $B_{ij}$ is a  Bernoulli RV with mean value $x_{ij}$. Thus, this observation follows from Observation~\ref{o:between-two-dist}.
\end{proof}

\begin{observation}
	\label{o:tickets-independent}
	For any job $j \in \cJ$, all RVs $\{N_{ij}\}_{i \in \cM}$ are independent. 
\end{observation}
\begin{proof}
	This is because  all RVs $\{\tilde N_{ij}\}_{i \in \cM, j \in \cJ}$ are independent; all RVs $\{B_{ij}\}_{i \in \cM}$ are independent; and $N_{ij} =  B_{ij} \cdot \tilde N_{ij}$.
\end{proof}

Recall that $\cJ_{< \ell}$ denotes the set of job that are assigned to some machines before the $\ell$-th iteration.

\begin{lemma}
	\label{lem:single}
	For any job $j \in \cJ$,  machine $i \in \cM$ and iteration $\ell \geq 1$, 
	we have $\Pr[ i \leftarrow_\ell j  \; | \; B^\ell_{ij}, j \notin \cJ_{< \ell}] = 1- 1/e$.
\end{lemma}
\begin{proof}
	Fix a job $j$ and machine $i$. Say $i = 1$ wlog By Observation~\ref{o:tickets-independent} and Fact~\ref{fact:poi}, we have $\sum_{i' \in \cM \setminus \{1\}} N_{i'j} \sim  \poi(\sum_{i'  \in \cM \setminus \{1\}} x_{i'j} = 1 - x_{1j})$. For notational convenience, let $\lambda = x_{1j}$, $\tilde N_1 = \tilde N_{1j}$, and $N_{-1} = \sum_{i' \in \cM \setminus \{1\}} N_{i'j}$. Note that $\tilde N_1 \sim \widetilde \poi(\lambda)$ and $N_{-1} \sim \poi(1-\lambda)$. 
\begin{align*}	
	&\;\;\; \Pr[ i \leftarrow_1  j \; | \; B_{ij}]  &\\
	&= \sum_{k \geq 1, k' \geq 0} \Pr[ \tilde N_1 = k \wedge N_{-1} = k']  \cdot \frac{k}{k+k'}  &\nonumber \\	
	&= \sum_{k \geq 1, k' \geq 0} \Pr[ \tilde N_1 = k] \cdot \Pr [N_{-1} = k']  \cdot \frac{k}{k+k'}   &\mbox{[$\tilde N_1$ and $N_{-1}$ are independent]} \\	
	&= \sum_{k \geq 1, k' \geq 0} e^{-\lambda}  \frac{\lambda^k}{\lambda k!}  \cdot e^{-(1- \lambda)} \frac{ (1 - \lambda)^{k'}}{ k'!}  \cdot \frac{k}{k+k'}  &  \mbox{[$\tilde N_1 \sim \widetilde \poi(\lambda)$ and $N_{-1} \sim \poi(1 - \lambda)$]}\\
	&= \sum_{k \geq 0, k' \geq 0} e^{-\lambda}  \frac{\lambda^k}{ k!}  \cdot e^{-(1- \lambda)} \frac{ (1 - \lambda)^{k'}}{k'!}  \cdot \frac{1}{k+k'+1} &\\		
	&= e^{-1} \sum_{k'' \geq 0} \sum_{k, k' \geq 0: k + k' = k''}  \frac{\lambda^k}{k!} \cdot \frac{(1 - \lambda)^{k'}}{k'!} \cdot k''! \cdot \frac{1}{(k''+1)!}  &\\
	&= e^{-1} \sum_{k'' \geq 0} \frac{1}{(k''+1)!}  (\lambda + (1 - \lambda))^{k''} &\\
	&= e^{-1} \sum_{k'' \geq 1} \frac{1}{(k'')!} 
%	&= e^{-1} \sum_{k'' \geq 0} e^{-1}  \frac{1}{ k''!}  \cdot \frac{1}{k''+1} \label{eqn:merge}\\			
	= e^{-1} (e - 1)  = 1 - 1/e, &
\end{align*}
\noindent which, combined with Observation~\ref{obs:indep-rounds}, yields the lemma. 
\end{proof}

\begin{corollary}
	\label{cor:second-1}
	For any job $j \in \cJ$, machine $i \in \cM$ and iteration $\ell \geq 1$, we have $\Pr[i \leftarrow_\ell j \;| \; j \notin \cJ_{< \ell}] = (1 - 1/e) x_{ij}$. 
\end{corollary}
\begin{proof}
	$\Pr[i \leftarrow_1 j] = \Pr[i \leftarrow_1 j \wedge B_{ij}] = \Pr[i \leftarrow_1 j \; | \; B_{ij}] \cdot \Pr[ B_{ij} ]  = (1 - 1/e) x_{ij}$. Then, we use Observation~\ref{obs:indep-rounds}.
\end{proof}

This also implies that job $j$ is assigned in the 1-st iteration with probability exactly $1 - 1/e$. 

\begin{corollary}
	\label{cor:second-2}
	For any job $j \in \cJ$ and iteration $\ell \geq 1$, we have $\Pr[nil \leftarrow_\ell j \; | \; j \notin \cJ_{< \ell}] = 1/e$. 
\end{corollary}
\begin{proof}
	By the above corollary and the fact that a job can be assigned to at most one machine in each iteration, the probability is 1 - $\sum_{i \in \cM} (1 - 1/e)x_{ij} = 1/e$. 	
\end{proof}

We are now ready to complete the proof of the second property. 
\begin{align*}	
	&\Pr[ i \leftarrow j]  =
\sum_{\ell \geq 1} \Pr[ i \leftarrow_\ell j]  &\\
= &\sum_{\ell \geq 1} \Pr[ i \leftarrow_\ell j \; | \; j \notin \cJ_{< \ell} ] \cdot \prod_{\ell' < \ell}  \Pr[nil \leftarrow_{\ell'} j \; | \; j \notin \cJ_{< \ell'}] &\\
 = &\sum_{\ell \geq 1}  (1 - 1/e) x_{ij} (1 /e^{\ell-1}) \\
= &x_{ij}&
\end{align*}

To complete the proof of the first property, we observe that 
$\Pr[ j \notin \cJ_{\leq \ell}] = \prod_{\ell' \leq \ell}  \Pr[nil \leftarrow_{\ell'} j \; | \; j \notin \cJ_{< \ell'}] = (1 / e)^\ell$ by Corollary~\ref{cor:second-2}. Thus, using the linearity of expectation, we know that the expected number of jobs remaining unassigned after $2 \log n$ iterations, i.e., $\E |\cJ \setminus \cJ_{\leq 2 \log n}| = n (1/e)^{2 \log n} = 1/ n$. Thus, by Markov inequality, the probability that the randomized rounding does not terminate within $2 \log n$ iterations, meaning that there is at least one job unassigned, is at most $1/n$. Also using a simple union bound we know $\Pr[ \cJ \neq \cJ_{\leq \ell}] \leq \min\{1, n  (1/e)^\ell\}$. Thus, the expected number of iterations before the termination is $\sum_{\ell \geq 1} \min\{1, n  (1/e)^\ell\} = O(\log n)$.  It is an easy exercise to see each iteration takes $O(mn)$ time assuming that we can determine the value of each RV in $O(1)$ time. As discussed already, it is clear that each job $j$ is assigned to exactly one machine if the algorithm terminates. 
Thus, we have shown the first property and running time.

%%%%%%%%%%%%%%%%%%%%%%%%%%%%%%%
%%%%%%%%%%%%%%%%%%%%%%%%%%%%%%%

\subsection{Third Property}  

We restate the third property we aim to prove: Fix a machine $i \in \cM$ and two distinct jobs $j, j' \in \cJ$ from \emph{different} groups of machine $i$, i.e., $j \in A$ and $j' \in B$ for some $A \neq B \in \cG_i$. Then, we have $Pr[ i \leftarrow j \wedge i \leftarrow j'] \leq  x_{ij} x_{ij'}$. For notational convenience, let $a = x_{ij}$ and $b = x_{ij'}$, and $\tilde M^* = \tilde N_{ij}$, $\tilde N^* = \tilde N_{ij'}$. Assume wlog that $a, b \in (0, 1)$ since otherwise the third property immediately follows from the second. 

Here is a very high-level overview of the proof. For the sake of contradiction assume that $Pr[ i \leftarrow j \wedge i \leftarrow j'] > ab$. Then, we will show that the probability remains greater than $ab$ after de-grouping the two jobs $j$ and $j'$ on all machines. This is a contradiction because 
two events $i \leftarrow j$ and $i \leftarrow j'$ are independent after de-grouping, which would immediately imply $Pr[ i \leftarrow j \wedge i \leftarrow j'] = ab$ due to the second property. 

However, the actual analysis is quite involved. To help the reader keep the flow of the analysis, we outline the proof in detail.

\subsubsection{Proof Outline}

 For the sake of analysis, we need to define additional notation---then, we can give a more detailed proof overview and explain the technical challenges. Let $M_{-} := \sum_{i'  \in \cM \setminus \{i\}} N_{i'j}$ denote the total number of \emph{real} tickets generated for job $j$ on machines other than $i$. Similarly, let $N_{-} := \sum_{i'  \in \cM \setminus \{i\}} N_{i'j'}$ for job $j'$.
Define $P_\cG(m, n) := Pr_\cG[ M_- = m, N_- = n]$; here $\cG$ in the subscript is to emphasize that this probability is under grouping $\cG$. Note that in this section we \emph{override} $m$ and $n$, which were used to denote the number of machines and jobs respectively in other sections. For comparison, create another grouping $\cG'$ by separating the two jobs $j$ and $j'$ in the same group on every machine. That is, for any machine $i'$ where $j, j' \in G$ for some $G \in \cG_{i'}$, partition the group $G$ arbitrarily into two groups $G_{ij}$ and $G_{ij'}$, so that $j \in G_{ij}$ and $j' \in G_{ij'}$. Let $P_{\cG'}(m, n)$ denote $\Pr_{\cG'}[M_- = m, N_- = n]$ 
under this grouping. 

Our goal is to show that 
\begin{equation}
	\label{eqn:third-outline}
\mbox{$\Pr_{\cG'}[  i \leftarrow j \wedge i\leftarrow j'] > ab \textnormal{ if  }\Pr_\cG[  i \leftarrow j \wedge i\leftarrow j'] > ab$}
\end{equation}
Unfortunately, proving this directly seems very challenging. Let's see why. First, our analysis has very little room for loss since we need to show negative correlation for ever pair of jobs not grouped together on the fixed machine. For careful analysis, we will have to take a close look at probabilities fixing some random variables. Then, when we expand $\Pr_\cG[  i \leftarrow j \wedge i\leftarrow j']$ depending on the iterations in which the jobs are assigned, we run across the recursive structure for the case $j, j' \notin J_{\leq 1}$. It seems very challenging to compare the above two probabilities in Eqn. (\ref{eqn:third-outline}) with this recursive structure combined with certain subtle conditions.

We get around this difficulty by first showing that  Eqn. (\ref{eqn:third-outline}) \emph{pretending} that after the first iteration, each unassigned job is assigned \emph{independently} in the second iteration. This thought process will ensure that the rounding terminates in two iterations, thus having no recursive structure -- let's  call this rounding as \emph{shadow} rounding. To relate the shadow rounding to the actual rounding, we will define $\Phi(P, \kappa, \tilde m^*, \tilde n^*)$. For notational convenience, let $\zeta := \Pr_{\cG}[  i \leftarrow j \wedge i\leftarrow j']  / (ab)$ and $\zeta' := \Pr_{\cG'}[  i \leftarrow j \wedge i\leftarrow j']  / (ab)$. Then, $\Phi$ will have the following nice properties.

\begin{enumerate}[label=(\alph*)]
	\item $\Phi(P_{\cG}, \zeta, \tilde m^*, \tilde n^*) = \Pr_\cG[  i \leftarrow j \wedge i\leftarrow j' \; | \; \tilde M^* = \tilde m^*, \tilde N^* = \tilde n^*]$ for our actual rounding.
	\item $\Phi(P_{\cG'}, \zeta', \tilde m^*, \tilde n^*) = \Pr_\cG[  i \leftarrow j \wedge i\leftarrow j' \; | \; \tilde M^* = \tilde m^*, \tilde N^* = \tilde n^*]$ for our actual rounding.
	\item $\Phi(P_{\cG}, 1, \tilde m^*, \tilde n^*) = \Pr_\cG[  i \leftarrow j \wedge i\leftarrow j' \; | \; \tilde M^* = \tilde m^*, \tilde N^* = \tilde n^*]$ for the shadow rounding.
	\item $\Phi(P_{\cG'}, 1, \tilde m^*, \tilde n^*) = \Pr_\cG[  i \leftarrow j \wedge i\leftarrow j' \; | \; \tilde M^* = \tilde m^*, \tilde N^* = \tilde n^*]$ for the shadow rounding.
\end{enumerate}

Depending on whether $\tilde m^* > 0$ or not and whether $\tilde n^* > 0$ or not, $\Phi$ is slightly different---so we will consider four cases, which are each presented in Sections~\ref{sec:third-pp}, \ref{sec:third-p0}, \ref{sec:third-0p} and \ref{sec:third-00}. We will show the above (a) and (b) for our actual rounding in Lemmas~\ref{lem:third-pp},~\ref{lem:third-p0},~\ref{lem:third-0p} and~\ref{lem:third-00}. We do not show (c) and (d) for the shadow rounding because our proof only uses the algebraic difference between $\Phi(P_{\cG}, \zeta, \tilde m^*, \tilde n^*)$ and  $\Phi(P_{\cG}, 1, \tilde m^*, \tilde n^*)$ and that between $\Phi(P_{\cG'}, \zeta', \tilde m^*, \tilde n^*)$ and $\Phi(P_{\cG'}, 1, \tilde m^*, \tilde n^*)$. Still, we mention (c) and (d) above  as we believe they could give more intuitions.

Then, we will show that 
$$	\mbox{$\Phi(P_{\cG'}, 1, \tilde m^*, \tilde n^*) \geq \Phi(P_{\cG}, 1, \tilde m^*, \tilde n^*)$}$$
for all $m^*, n^* \geq 0$ in Lemmas~\ref{lem:third-pp-mono},~\ref{lem:third-p0-mono},~\ref{lem:third-0p-mono} and~\ref{lem:third-00-mono}.
By de-conditioning on $\tilde M^*$ and $\tilde N^*$ and using the aforementioned algebraic difference, we will be able to show 
	Eqn. (\ref{eqn:third-outline}).
We will present the details on how to put all the pieces together in Section~\ref{sec:third-together}. This will complete the proof of the third property.

Due to the space constraints, we defer to the full version of this paper the proof of each lemma.

\subsubsection{Case $\tilde M^*, \tilde N^* > 0$}
	\label{sec:third-pp}

Recall $M_{-} := \sum_{i'  \in \cM \setminus \{i\}} N_{i'j}$ denotes the total number of \emph{real} tickets generated for job $j$ on machines other than $i$ and  $N_{-}$ is similarly defined for job $j'$. Also recall $P(m, n) := Pr[ M_- = m, N_- = n]$. 
\begin{lemma}
	\label{lem:third-pp}
For every $\tilde m^*, \tilde n^* \geq 1$, define:
\begin{align}
\Phi(P, \kappa, \tilde m^*, \tilde n^*) := & ab \hspace{-1.5ex} \sum_{m \geq 0, n \geq 0} \frac{\tilde m^*}{m + \tilde m^*} \cdot \frac{\tilde n^*}{n+ \tilde n^*} \cdot P(m, n) \label{eqn:3-1}\\
	&+  ab  ~ \sum_{m \geq 0} \frac{\tilde m^*}{m + \tilde m^*} \cdot (1 - b) \cdot P(m, 0)   \label{eqn:3-2} \\	
	&+  ab  ~ \sum_{n \geq 0}  \frac{\tilde n^*}{n+ \tilde n^*} \cdot (1 - a) \cdot P(0, n)  \label{eqn:3-3} \\		
	&+  ab  ~ (1 -a) ( 1 - b) \kappa \cdot P(0, 0) \label{eqn:3-4}
\end{align}
Then, we have 
\begin{itemize}
\item $\Phi(P_{\cG}, \zeta, \tilde m^*, \tilde n^*) = \Pr_\cG[ i \leftarrow j \wedge i\leftarrow j' \;| \; \tilde M^* = \tilde m^*, \tilde N^* = \tilde n^*]$; and
\item $\Phi(P_{\cG'}, \zeta', \tilde m^*, \tilde n^*) = \Pr_{\cG'}[ i \leftarrow j \wedge i\leftarrow j' \;| \; \tilde M^* = \tilde m^*, \tilde N^* = \tilde n^*]$.
\end{itemize}
\end{lemma}
\begin{proof}
	Note that Eqn. (\ref{eqn:3-1}), (\ref{eqn:3-2}), (\ref{eqn:3-3}) have no dependency on  $\kappa$. 
 	Eqn. (\ref{eqn:3-1}) is for event $(i \leftarrow_1 j \wedge i \leftarrow_1 j')$ since
\begin{align*}
	& \Pr[i \leftarrow_1 j \wedge i \leftarrow_1 j' \; | \; \tilde M^* = \tilde m^*, \tilde N^* = \tilde n^*, M_- = m , N_- = n]  \\
	= &\frac{\tilde m^*}{m + \tilde m^*} \cdot \Pr[B_{ij}] \cdot \frac{\tilde n^*}{n+ \tilde n^*} \cdot \Pr[B_{ij'}]\\  
	= & ab ~ \frac{\tilde m^*}{m + \tilde m^*} \cdot \frac{\tilde n^*}{n+ \tilde n^*}.
\end{align*}
Here we used the fact that $B_{ij}$ and $B_{ij'}$ are independent as $j$ and $j'$ are not in the same group on machine $i$. By de-conditioning on $M_-$ and $N_-$ using the fact that $(\tilde M^*,  \tilde N^*)$ is independent of $(M_-, N_-)$, we have Eqn. (\ref{eqn:3-1}).

	 	Eqn. (\ref{eqn:3-2}) is for event $(i \leftarrow_1 j \wedge i \leftarrow_{\geq 2} j')$. Note that this event occurs only 	
	when $\neg B_{ij'}$ (since $\tilde N^*$ is fixed to a non-zero value) and $m = 0$.  Formally, we have, 
\begin{align*}
	&\;\;\; \Pr[i \leftarrow_1 j \wedge nil \leftarrow_1 j' \; | \; \tilde M^* = m^*, \tilde N^* = \tilde n^*, M_- = m, N_- = n = 0] \\
	&= \frac{\tilde m^*}{m + \tilde m^*} \cdot \Pr[B_{ij}] \cdot \Pr[ \neg B_{ij'}]   
	=  \frac{\tilde m^*}{m + \tilde m^*} \cdot a(1-b)
\end{align*}
 By de-conditioning on $M_-$ and $N_-$, we have
\begin{align*}
	&\;\;\; \Pr[i \leftarrow_1 j \wedge nil \leftarrow_1 j' \; | \; \tilde M^* = \tilde m^*, \tilde N^* = \tilde n^*]  
=  \sum_{m \geq 0} \frac{\tilde m^*}{m + \tilde m^*} \cdot a(1-b)\cdot P(m, 0)
\end{align*}

Due to the second property and Observation~\ref{obs:indep-rounds}, we have  Eqn. (\ref{eqn:3-2}).

Since Eqn. (\ref{eqn:3-3}) is for event $(i \leftarrow_{\geq 2} j \wedge i \leftarrow_{1} j')$, which is symmetric to the event of Eqn. (\ref{eqn:3-2}), it remains to show Eqn. (\ref{eqn:3-4}). We will only show Eqn. (\ref{eqn:3-4}) with  $\kappa = \zeta$ for grouping $\cG$, which is for  event $(i \leftarrow_{> 1} j' \wedge i \leftarrow_{ > 1} j')$, since we can similarly show Eqn. (\ref{eqn:3-4}) with  $\kappa = \zeta'$ for grouping $\cG'$.
Note that this event occurs if and only if $\neg B_{ij}$, $\neg B_{ij'}$, $m = n= 0$ and $j$ and $j'$ are assigned to $i$ in the subsequent iterations. 
Thus, by Observation~\ref{obs:indep-rounds}, the event occurs with probability $P(0, 0) \cdot \Pr[\neg B_{ij} \wedge \neg B_{ij'}] \cdot \Pr[i \leftarrow j \wedge i \leftarrow j'] = (1 - a) (1 - b) P(0, 0) \Pr[i \leftarrow j \wedge i \leftarrow j'] =  (1 -a) (1- b) P(0, 0) \zeta ab$, as desired. 
\end{proof}

In the following we re-write $\Phi$, so that we have the  co-efficient of each $P(m, n)$ explicitly.

\begin{lemma}
For any $\tilde m^*, \tilde n^* \geq 1$, we have 	$\Phi(P, \kappa, \tilde m^*, \tilde n^*) = \sum_{m, n \geq 0} d(m, n) P(m, n)$, where \\
$\displaystyle
d(m, n) := 
\begin{cases} 
ab \frac{\tilde m^*}{m + \tilde  m^*} \cdot \frac{\tilde n^*}{n+ \tilde n^*} &\mbox{if } m, n > 0\\ 
ab \frac{\tilde m^*}{m + \tilde m^*} (2 - b) &\mbox{if } m >0, n = 0\\ 
ab \frac{\tilde n^*}{n + \tilde n^*} (2- a) &\mbox{if } m = 0, n > 0\\  
ab (2  - a) (2   - b) + (\kappa -1) (1 -a) (1 - b) &\mbox{if } m =  n = 0;  
\end{cases}
$
\end{lemma}
\begin{proof}
\begin{align*}
	& \frac{\Phi(P, \kappa, \tilde m^*, \tilde n^*)}{ab} \\
	= &  \sum_{m > 0, n > 0} \frac{\tilde m^*}{m + \tilde m^*} \cdot \frac{\tilde n^*}{n+ \tilde n^*} \cdot P(m, n) +  \sum_{m > 0} \frac{\tilde m^*}{m + \tilde m^*}  \cdot P(m, 0)  +  \sum_{ n > 0}  \frac{\tilde n^*}{n+ \tilde n^*} \cdot P(0, n) +  P(0, 0)\\
& +  \sum_{m > 0} \frac{\tilde m^*}{m + \tilde m^*}  \cdot P(m, 0)  (1 - b) +     P(0, 0)  (1 - b)  \\	
 & +  \sum_{n > 0}  \frac{\tilde  n^*}{n+ \tilde  n^*} \cdot P(0, n) (1 - a) +       P(0, 0) (1 - a)  \\		
	 & +   P(0, 0) (1 -a) (1 - b) + P(0, 0) (\kappa -1) (1 -a) (1 - b) \\
	= &  \sum_{m > 0, n > 0} \frac{\tilde m^*}{m + \tilde m^*} \cdot \frac{\tilde n^*}{n+ \tilde n^*} \cdot P(m, n) \\
	  &+ \sum_{m > 0} \frac{\tilde m^*}{m + \tilde m^*}  \cdot P(m, 0)  (2- b)  
	  +  \sum_{n > 0}  \frac{\tilde n^*}{n+ \tilde n^*} \cdot P(0, n) (2 - a)  \\
	  &+  P(0, 0) ( (2- a)   (2- b) +(\kappa -1) (1 -a) (1 - b) ) 
\end{align*}
\end{proof}

\begin{lemma}
	\label{lem:third-pp-mono}
	For any $\tilde m^* , \tilde n^* \geq 1$, $\Phi(P_\cG, 1, \tilde m^*,  \tilde n^*) \leq \Phi(P_{\cG'}, 1, \tilde m^*,  \tilde n^*)$.
\end{lemma}
\begin{proof}
	 Fix $\tilde M^* = \tilde m^* , \tilde N^* =  \tilde n^*$ where $\tilde m^*, \tilde n^* \geq 1$. 
	For notational convenience, let machine 1 ($\neq i$) be the unique machine where $\cG$ and $\cG'$ differ, i.e., $\cG_{1} \neq \cG_{1}$ but $\cG_{i'} = \cG'_{i'}$ for all $i' \in \cM \setminus \{1\}$. So, $j$ and $j'$ are in the same group of $\cG_1$ but are in different groups of $\cG'_1$.
	 
	 Let $P_{-1}$ denote $P_\cG$ (or equivalently $P_{\cG'}$) without counting the tickets from machine 1. Formally, $P_{-1}:= \sum_{i'  \in \cM \setminus \{1, i\}} N_{i'j}$. We will compare the outcome of the randomized rounding under $\cG$ and $\cG'$.  Also fix $\tilde u := \tilde N_{1j}$ and $\tilde  v := \tilde N_{1j'}$ ---note that this is independent of any grouping. Let $a_1 : = x_{1j}$ and $b_1 := x_{1j'}$ for notational convenience.  Assume $a_1, b_1 \in (0,1)$ since otherwise $B_{1j}$ and $B_{1j'}$ will be independent under both $\cG$ and $\cG'$
	 
	  Note that 
  	$$P_{\cG}(m, n) := P_{-1}(m, n) (1 - a_1 - b_1) + P(m - \tilde u, n) a_1  + P(m, n - \tilde  v)  b_1 $$
\noindent because the $\tilde u$ ($\tilde  v$, resp.) potential tickets become real when $B_{1j}$ ($B_{1j'}$, resp.), which occur with probability $a_1$ ($b_1$, resp.). Here, $P(m, n) := 0$ if $m < 0$ or $n < 0$.  With another grouping $\cG'$, we have  
\begin{align*}
  P_{\cG'}(m, n) := & P_{-1}(m, n) (1 - a_1) (1 - b_1) \\&+ P_{-1}(m - \tilde  u, n) a_1 (1 - b_1) + P_{-1}(m, n - \tilde  v) (1 - a_1) b_1 \\&+ P_{-1}(m - \tilde  u, n - \tilde  v) a_1 b_1  
\end{align*}
	
\noindent because $B_{1j}$ and $B_{1j'}$ are independent under $\cG'$.

	Our goal is to show that $\Phi(P_\cG, 1, \tilde m^*,  \tilde n^*) \leq \Phi(P_{\cG'}, 1, \tilde m^*,  \tilde n^*)$ (for any fixed $\tilde M^*$, $\tilde N^*$, $\tilde  u$ and $\tilde  v$).  
	Using the observation that $P_{\cG'}(m, n) - P_{\cG}(m, n)  = a_1b_1 \Big(  P_{-1}(m, n)  - P_{-1}(m - \tilde u, n) - P_{-1}(m, n- \tilde v)  + P_{-1}(m- \tilde u, n- \tilde v) \Big)$, we derive,
	\begin{align*}
	&  \Phi(P_{\cG'}, 1, \tilde m^*,  \tilde n^*) - 	\Phi(P_\cG, 1, \tilde m^*,  \tilde n^*) \\
	= 	&\sum_{m, n \geq 0}  d(m, n) \Big( P_{\cG'}(m, n)   - P_{\cG}(m, n) \Big) \\
	= 	&a_1b_1 \sum_{m, n \geq 0} d(m, n) \Big(  P_{-1}(m, n)  - P_{-1}(m - \tilde u, n) - P_{-1}(m, n- \tilde v)+ P_{-1}(m- \tilde u, n- \tilde v) \Big) \\	
	= 	&a_1b_1 \sum_{m, n \geq 0} d(m, n) P_{-1}(m, n)  -  a_1b_1 \sum_{m, n \geq 0} d(m + \tilde u , n) P_{-1}(m, n) \\
	&\;\;\;-   a_1b_1  \sum_{m, n \geq 0} d(m  , n+ \tilde v) P_{-1}(m, n) + a_1b_1   \sum_{m, n \geq 0} d(m + \tilde u , n+ \tilde v) P_{-1}(m, n)   \\
	= 	&a_1b_1  \sum_{m, n \geq 0} P_{-1}(m, n) \Big( d(m, n) - d(m+\tilde  u, n)  - d(m, n+ \tilde v)+ d(m+\tilde u, n+ \tilde v) \Big) 
	\end{align*}
	
	Therefore, to establish the lemma, it suffices to show that 
	$$D := d(m, n) - d(m+ \tilde u, n)  - d(m, n+ \tilde v) + d(m+ \tilde u, n+ \tilde v) \geq 0 \textnormal{ for all } m, n, \tilde u, \tilde v \geq 0$$ 
Observe that the claim is immediate if $\tilde u = 0$ or $\tilde v = 0$, so assume that $\tilde u, \tilde v > 0$.  Let $f(m) := \frac{\tilde m^*}{m + \tilde m^*}$ and $g(n) := \frac{\tilde n^*}{n+ \tilde n^*}$, which are both decreasing in $m$ and $n$, respectively. The first case we consider is when  $m, n > 0$. Then, we have 
$$\frac{D}{ab} = f(m) g(n) - f(m +\tilde u) g(n) - f(m) g(n+\tilde  v) + f(m+\tilde  u) g(n+\tilde  v) = (f(m) - f(m+\tilde  u)) (g(n) - g(n+\tilde  v)) > 0.$$

In the second case when $m > 0$ and $n= 0$,  we have, 
\begin{align*}
D / (ab)	&= (2- b) ( f(m) - f(m +\tilde  u) )  - f(m) g(n+\tilde  v) + f(m+\tilde  u) g(n+\tilde  v) \\
	&= (2 - b) ( f(m) - f(m +\tilde  u) )  - g(n+\tilde  v) (f(m) - f(m+ \tilde  u)) \\
	&= (2- b - g(n+\tilde v)) ( f(m) - f(m +\tilde  u) ) \geq 0, 
\end{align*}
\noindent since $g(n +v) \leq 1$ and $f$ is decreasing in $m$. 
The case $m = 0$ and $n >0$ is symmetric. In the last case when $m = n = 0$, we have, 
\begin{align*}
D /(ab)	&\geq (2 - b)   (2 - a) - (2 - a) g(\tilde v)  - (2 - b) f(\tilde  u)  - f(\tilde  u) g(\tilde  v) \\
	&= (2- a - f(\tilde  u)) (2- b- g(\tilde v)) \geq 0, 	
\end{align*}
\noindent since $a, b, f(\tilde u), g(\tilde v) \leq 1$.
\end{proof}

\subsubsection{Case $\tilde M^* >0,  \tilde N^* = 0$} 
	\label{sec:third-p0}

Previously, we considered the case $\tilde M^*, \tilde N^* > 0$. Here, we consider the case $\tilde M  = m^*, \tilde N^*= 0$ for any fixed $m^* \geq 1$. 

\begin{lemma}
	\label{lem:third-p0}
For any $\tilde m^* > 0, \tilde n^* = 0$, define,
\begin{align*} 
\Phi(P, \kappa, \tilde m^*, \tilde n^* = 0) &: =  \sum_{m \geq 0}  d(m, 0)  P(m, 0) \textnormal{ where } 
\end{align*}
$$
d(m, n) = 
\begin{cases} 
ab \frac{\tilde m^*}{m + \tilde  m^*} &\mbox{if } m > 0, n = 0\\ 
ab  (2 - a) + ab (\kappa - 1)(1 - a) &\mbox{if } m =0, n = 0\\
0 &\mbox{if } n > 0
\end{cases}
$$
Then, we have 
\begin{itemize}
	\item $\Phi(P, \zeta, \tilde m^*, 0) = Pr_{\cG}[ i \leftarrow j \wedge i\leftarrow j' \;| \; \tilde M^* = \tilde m^*, \tilde N^* = 0]$; and 
	\item $\Phi(P, \zeta', \tilde m^*, 0) = Pr_{\cG'}[ i \leftarrow j \wedge i\leftarrow j' \;| \; \tilde M^* = \tilde m^*, \tilde N^* = 0]$.
\end{itemize}
\end{lemma}
\begin{proof}
	We only consider grouping $\cG$ as the other grouping $\cG'$ can be handled analogously. 	Conditioned on $\tilde N^* = 0$, we have $\neg (i \leftarrow_1 j')$. Therefore, we can refine $i \leftarrow j \wedge i\leftarrow j'$ into two disjoint events: $(i \leftarrow_1 j \wedge i\leftarrow_{\geq 2} j')$ and $(i \leftarrow_{\geq 2} j \wedge i\leftarrow_{\geq 2} j')$. Note that  $i\leftarrow_{\geq 2} j'$ only if $N_- = 0$. We consider the first event: 
\begin{align*}	
&\;\;\;\;\;	\Pr[i \leftarrow_1 j \wedge i\leftarrow_{\geq 2} j' \; | \; \tilde M^* = \tilde m^*, \tilde N^* = 0, M_- = m, N_- = 0] \\
&	= \Pr[ i\leftarrow_{\geq 2} j'  \; | \; i \leftarrow_1 j, \tilde M^* = \tilde m^*, \tilde N^* = 0, M_- = m, N_-= 0]\\
&\;\;\;\;\;\;\;\;\; \cdot \Pr[i \leftarrow_1 j \; | \; \tilde M^* = \tilde m^*, \tilde N^* = 0, M_- = m, N_-= 0] \\
&	= \Pr[ i\leftarrow_{\geq 1} j'] \cdot \Big (a\frac{\tilde m^*}{m + \tilde m^*}\Big) = ab \frac{\tilde m^*}{m + \tilde m^*}
\end{align*}
The penultimate equality is due to Observation~\ref{obs:indep-rounds}. By de-conditioning on $M_-$ and $N_-$, we have 
\begin{equation}
	\label{eqn:third-1}
\Pr[i \leftarrow_1 j \wedge i\leftarrow_{\geq 2} j' \;| \; \tilde M^* = \tilde m^*, \tilde N^* = 0] = 	ab \sum_{m \geq 0}  \frac{\tilde m^*}{m + \tilde m^*} P(m, 0)
\end{equation}

	We now consider the second event. Notice that conditioned on $\tilde M^* > 0$, $nil \leftarrow_1  j$  occurs iff $\neg B_{ij}$ and $M_- = 0$. Likewise, conditioned on $\tilde N^* = 0$, $nil \leftarrow_1 j'$ occurs iff $N_- = 0$. Thus, 
\begin{align}	
	&\Pr[i \leftarrow_{\geq 2} j \wedge i\leftarrow_{\geq 2} j' \; | \; \tilde M^* = \tilde m^*, \tilde N^* = 0]  \nonumber\\
	=&\Pr[i \leftarrow_{\geq 2} j \wedge i\leftarrow_{\geq 2} j' \wedge \neg B_{ij} \wedge M_- = N_- = 0 \; | \; \tilde M^* = \tilde m^*, \tilde N^* = 0]  \nonumber\\
		     = &	 \Pr[ i\leftarrow_{\geq 2} j'  \wedge i \leftarrow_{\geq 2} j \; | \; \neg B_{ij} , M_- = N_- = 0, \tilde M^* = \tilde m^*, \tilde N^* = 0]  \nonumber \\
		     &\;\;\;\cdot \Pr[\neg B_{ij} \wedge M_- = N_-  = 0 \; | \; \tilde M^* = \tilde m^*, \tilde N^* = 0] \nonumber\\
     	= &  \Pr[ i\leftarrow_{\geq 1} j'   \wedge i \leftarrow_{\geq 1} j]  \cdot (1- a) P(0, 0)	\nonumber \\
	= & ab \zeta (1- a) P(0, 0)	 	\label{eqn:third-2}
\end{align}

	The penultimate equality follows due to the fact that $ \neg B_{ij} \wedge M_- = N_- = 0, \tilde M^* = \tilde m^*, \tilde N^* = 0$ implies $j, j' \notin J_{\leq 1}$ and  Observation~\ref{obs:indep-rounds}. By adding up Eqn. (\ref{eqn:third-1}) and (\ref{eqn:third-2}) and rearranging terms, we have the lemma. 
\end{proof}

\begin{lemma}
	\label{lem:third-p0-mono}
	For all $\tilde m^*  \geq 1$,  we have $\Phi(P_\cG, 1, \tilde m^*, \tilde n^*) \leq \Phi(P_{\cG'}, 1, \tilde m^*, \tilde n^*)$.
\end{lemma}
\begin{proof}
	The proof is very similar to that of Lemma~\ref{lem:third-pp-mono}.
	As before, it suffices to show that 
	$$D:= d(m, n) - d(m+ \tilde u, n)  - d(m, n+ \tilde v) + d(m+ \tilde u, n+ \tilde v) \geq 0$$
\noindent for $m, n, \tilde  u, \tilde v \geq 0$; but with $d(m, n)$ defined in  Lemma~\ref{lem:third-p0}.

Assume $\tilde u, \tilde v > 0$ since otherwise $D= 0$ immediately. Also assume $n = 0$; otherwise $D = 0$. So, we have $D= d(m, 0) - d(m+ \tilde u, 0)$. We consider two cases. If $m > 0$, then $D / (ab) = \frac{\tilde m^*}{m + \tilde  m^*} - \frac{\tilde m^*}{m + \tilde u + \tilde  m^*} \geq 0$. If $m = 0$, 
we have $D / (ab) =  (2 - a) + (1- 1)(1 - a)  - \frac{\tilde m^*}{ \tilde u + \tilde  m^*} =  2 - a - \frac{\tilde m^*}{ \tilde u + \tilde  m^*}  \geq 0$.
\end{proof}

\subsubsection{Case $\tilde M^*  = 0 ,  \tilde N^* >0$} 
	\label{sec:third-0p}

This case is symmetric to $\tilde M^* >0,  \tilde N^* = 0$. So, we just state the definition and lemma without proof. 

\begin{lemma}
	\label{lem:third-0p}
For any $\tilde n^* > 0$, define,
\begin{align*} 
\Phi(P, \kappa,  \tilde m^* = 0, \tilde n^*) &: =  \sum_{n \geq 0}  d(0, n)  P(0, n) \textnormal{ where } 
\end{align*}
$$
d(m, n) = 
\begin{cases} 
ab \frac{\tilde n^*}{n + \tilde  n^*} &\mbox{if } m = 0, n > 0\\ 
ab  (2 - b) + ab (\kappa - 1)(1 - b) &\mbox{if } m =0, n = 0\\
0 &\mbox{if } m > 0
\end{cases}
$$
Then, we have 
\begin{itemize}
	\item $\Phi(P, \zeta, 0, \tilde n^*) = Pr_{\cG}[ i \leftarrow j \wedge i\leftarrow j' \;| \; \tilde M^* = 0, \tilde N^* = \tilde n^*]$; and 
	\item $\Phi(P, \zeta', 0, \tilde n^*) = Pr_{\cG'}[ i \leftarrow j \wedge i\leftarrow j' \;| \; \tilde M^* = 0, \tilde N^*= \tilde n^*]$.
\end{itemize}
\end{lemma}

\begin{lemma}
	\label{lem:third-0p-mono}
	For all $\tilde n^*  \geq 1$,  $\Phi(P_\cG, 1, 0, \tilde n^*) \leq \Phi(P_{\cG'}, 1, 0, \tilde n^*)$.
\end{lemma}

\subsubsection{Case $\tilde M^* = \tilde N^* = 0$} 
	\label{sec:third-00}

This is the last case we consider. 

\begin{lemma}
	\label{lem:third-00}
	Define $\Phi(P, \kappa, 0, 0) := ab \kappa P(0, 0)$.
\begin{itemize}
	\item $\Phi(P, \zeta, 0, 0) = Pr_{\cG}[ i \leftarrow j \wedge i\leftarrow j' \;| \; \tilde M^* = 0, \tilde N^* = 0]$; and 
	\item $\Phi(P, \zeta', 0, 0) = Pr_{\cG'}[ i \leftarrow j \wedge i\leftarrow j' \;| \; \tilde M^* = 0, \tilde N^*= 0]$.
\end{itemize}
\end{lemma}
\begin{proof}
	We only consider grouping $\cG$ as the other grouping $\cG'$ can be handled analogously.
	 For $i \leftarrow j \wedge i\leftarrow j'$ to happen conditioned on neither $j$ nor $j'$ having potential (therefore real) tickets on machine $i$, it must be the case that neither $j$ nor $j'$ are assigned in the first iteration. So, it must be the case that $M_- = N_- = 0$. Therefore, we have 
	 $\Pr[ j, j' \notin J_{\leq 1} \;| \; \tilde M^* =  \tilde N^* = 0] = P(0, 0)$.
	 Further, we know that $\Pr[ i \leftarrow_{\geq 2} j \wedge i \leftarrow_{\geq 2} j'  \;| \; j, j' \notin J_{\leq 1},  \tilde M^* =  \tilde N^* = 0] = \Pr[ i \leftarrow_{\geq 1} j \wedge i \leftarrow_{\geq 1} j'] = \zeta ab$ by Observation~\ref{obs:indep-rounds}. Thus, we have the lemma.
\end{proof}

\begin{lemma}
	\label{lem:third-00-mono}
	$\Phi(P_\cG, 1, 0, 0) \leq \Phi(P_{\cG'}, 1, 0, 0)$.
\end{lemma}
\begin{proof}
As in the proof of Lemma~\ref{lem:third-pp-mono}, it suffices to show that
		$$D:= d(m, n) - d(m+ \tilde u, n)  - d(m, n+ \tilde v) + d(m+ \tilde u, n+ \tilde v) \geq 0$$
\noindent for all $m, n, \tilde  u, \tilde v \geq 0$; but with $d$ defined in Lemma~\ref{lem:third-00}. 

To see this, assume $\tilde u, \tilde v >0$, since otherwise $D = 0$. Further, assume $m = n = 0$ since otherwise $D = 0$. Then, we have $D = d(m, n)  = ab \cdot 1 = ab > 0$. 
\end{proof}

\subsubsection{Putting All Pieces Together} 
	\label{sec:third-together}

From the definitions stated in Lemmas~\ref{lem:third-pp},~\ref{lem:third-p0},~\ref{lem:third-0p} and~\ref{lem:third-00}, we have :
\begin{equation}
	\label{eqn:third-together-1}
\Phi(P _{\cG}, 1, \tilde m^*, \tilde n^*) - \Phi(P _{\cG}, \zeta, \tilde m^*, \tilde n^*) = 
\begin{cases} 
- (1 - a) (1 - b) (\zeta - 1) ab P_{\cG}(0, 0)  &\mbox{if } \tilde m^* > 0, \tilde n^* >0 \\ 
-  (1 - a)  (\zeta - 1) ab P_{\cG}(0, 0)  &\mbox{if } \tilde m^* > 0, \tilde n^* = 0 \\ 
 -  (1 - b)  (\zeta - 1)ab  P_{\cG}(0, 0) &\mbox{if } \tilde m^* = 0, \tilde n^* > 0 \\ 
 - (\zeta - 1) ab P_{\cG}(0, 0)   &\mbox{if } \tilde m^* = 0, \tilde n^* = 0 \\ 
\end{cases}
\end{equation}

Similarly, we have 
\begin{equation}
	\label{eqn:third-together-2}
\Phi(P _{\cG'}, 1, \tilde m^*, \tilde n^*) - \Phi(P _{\cG'}, \zeta', \tilde m^*, \tilde n^*) = 
\begin{cases} 
- (1 - a) (1 - b) (\zeta' - 1) ab P_{\cG'}(0, 0)  &\mbox{if } \tilde m^* > 0, \tilde n^* >0 \\ 
-  (1 - a)  (\zeta' - 1) ab P_{\cG'}(0, 0)  &\mbox{if } \tilde m^* > 0, \tilde n^* = 0 \\ 
 -  (1 - b)  (\zeta' - 1)ab  P_{\cG'}(0, 0) &\mbox{if } \tilde m^* = 0, \tilde n^* > 0 \\ 
 - (\zeta' - 1) ab P_{\cG'}(0, 0)   &\mbox{if } \tilde m^* = 0, \tilde n^* = 0 \\ 
\end{cases}
\end{equation}

In Lemmas~\ref{lem:third-pp-mono},~\ref{lem:third-p0-mono},~\ref{lem:third-0p-mono} and~\ref{lem:third-00-mono} we have shown that  $\Phi(P_{\cG}, 1, \tilde m^*, \tilde n^*) \leq   \Phi(P_{\cG'}, 1,  \tilde m^*, \tilde n^*)$ for all $\tilde m^*, n^* \geq 0$. Therefore, we have, 
$$
\sum_{\tilde m^*, \tilde n^* \geq 0} q(\tilde m^*, \tilde n^*) \Phi(P_{\cG}, 1, \tilde m^*, \tilde n^*) 
\leq  \sum_{\tilde m^*, \tilde n^* \geq 0} q(\tilde m^*, \tilde n^*) \Phi(P_{\cG'}, 1,  \tilde m^*, \tilde n^*), 
$$
\noindent where $q(\tilde m^*, \tilde n^*) := \Pr[ \tilde M^*  = \tilde m^*, \tilde N^* = \tilde n^*]$.

Using Eqn. (\ref{eqn:third-together-1}), we have, 
\begin{align*}
&\sum_{\tilde m^*, \tilde n^* \geq 0} q(\tilde m^*, \tilde n^*) \Phi(P_{\cG}, 1, \tilde m^*, \tilde n^*)  \\
=&\sum_{\tilde m^*, \tilde n^* \geq 0} q(\tilde m^*, \tilde n^*)  \Phi(P _{\cG}, \zeta, \tilde m^*, \tilde n^*) \\
&- \sum_{\tilde m^*, \tilde n^* > 0}  q(\tilde m^*, \tilde n^*) (1 - a) (1 - b) (\zeta - 1) ab P_{\cG}(0, 0) \\
&- \sum_{\tilde m^* >0 } q(\tilde m^*, 0)  (1 - a)  (\zeta - 1) ab P_{\cG}(0, 0) \\
& -\sum_{\tilde n^* >0 }  q(\tilde 0, \tilde n^*)  (1 - b)  (\zeta - 1)ab  P_{\cG}(0, 0)   \\
& - q(0, 0) (\zeta - 1) ab P_{\cG}(0, 0)   \\
= & \Pr_\cG[i \leftarrow j, i \leftarrow j'] - ab P_\cG(0, 0) (\zeta - 1) \Big( (1 - a) (1 - b)   q_{++} + (1 - a) q_{+0} + (1 - b) q_{0+}  +q_{00}\Big) \\
= & \zeta ab - ab P_\cG(0, 0) (\zeta - 1) \Big( (1 - a) (1 - b)   q_{++} + (1 - a) q_{+0} + (1 - b) q_{0+}  +q_{00}\Big),
\end{align*}
\noindent where $q_{++} := \Pr[ \tilde M^* >0, \tilde N^* >0], q_{+0} := \Pr[ \tilde M^* >0, \tilde N^* =0], q_{0+} := \Pr[ \tilde M^* = 0, \tilde N^* >0], q_{00} := \Pr[ \tilde M^* = 0, \tilde N^* =0]$.

Similarly, we obtain 
\begin{align*}
&\sum_{\tilde m^*, \tilde n^* \geq 0} q(\tilde m^*, \tilde n^*) \Phi(P_{\cG'}, 1,  \tilde m^*, \tilde n^*)  \\
= & \zeta'  ab - ab P_{\cG'}(0, 0) (\zeta' - 1) \Big( (1 - a) (1 - b)   q_{++} + (1 - a) q_{+0} + (1 - b) q_{0+}  +q_{00}\Big)
\end{align*}

Thus, we have, 
\begin{align}
   \zeta  ab - ab P_{\cG}(0, 0) (\zeta - 1) \Big( (1 - a) (1 - b)   q_{++} + (1 - a) q_{+0} + (1 - b) q_{0+}  +q_{00}\Big) \label{third-final-1} \\
\leq   \zeta'  ab - ab P_{\cG'}(0, 0) (\zeta' - 1) \Big( (1 - a) (1 - b)   q_{++} + (1 - a) q_{+0} + (1 - b) q_{0+}  +q_{00}\Big) \label{third-final-2}
\end{align}

Note that Eqn. (\ref{third-final-1}) and (\ref{third-final-2}) are linear  $\zeta$ and $\zeta'$, respectively. Further, $q_{++}+ q_{+0} + q_{0+} + q_{00} = 1$ and $q_{++} >0$ since $a, b \in (0, 1)$. Therefore, both $\zeta$ and $\zeta'$ have strictly positive coefficients in
Eqn. (\ref{third-final-1}) and (\ref{third-final-2}), respectively. Further, Eqn. (\ref{third-final-1}) has value $ab$ when $\zeta = 1$ and Eqn. (\ref{third-final-2}) has value $ab$ when $\zeta' =1$. Thus, if $\zeta >1$, then it must be the case that $\zeta' > 1$.

To summarize, we have shown that if $\Pr_{\cG}[i \leftarrow j \wedge i \leftarrow j'] > x_{ij} x_{ij'}$, then it must be the case that $\Pr_{\cG'}[i \leftarrow j \wedge i \leftarrow j'] > x_{ij} x_{ij'}$. Since $\cG'$ was  obtained from $\cG$ by de-grouping the two jobs $j$ and $j'$ on one machine, by repeatedly refining $\cG'$ further, we know that what we have shown still holds true when $\cG'$ doesn't group $j$ and $j'$ together on any machines. As mentioned before, this implies that $i \leftarrow j$ and $i \leftarrow j'$ are independent under $\cG'$. Thus, we have $\Pr_{\cG'}[i \leftarrow j \wedge i \leftarrow j'] = x_{ij} x_{ij'}$. Thus, if we assume  $\Pr_{\cG}[i \leftarrow j \wedge i \leftarrow j'] > x_{ij} x_{ij'}$, then we obtain a contradiction. This completes the proof of the third property.

%%%%%%%%%%%%%%%%%%%%%%%%%%%%%%%%%%%%%%%%%%%%%%%%%%%%%

%\bigskip
%\noindent
%\textbf{Fourth Property}. 

\subsection{Fourth Property}

For ease of reference, we re-state the fourth property we aim to prove: For every $i \in \cM$ and $j \neq j' \in \cJ$ such that  $j, j' \in G$ for some $G \in \cG_i$, we have $\Pr[i \leftarrow j \wedge i \leftarrow j']  \leq \frac{1}{1+e} (e^{x_{ij}} + e^{x_{ij'}})x_{ij}x_{ij'}$.

We first upper bound the probability that neither $j$ nor $j'$ are assigned in the first iteration.   Towards this end, we need the following proposition. 

\begin{proposition}
	\label{pro:conditional-num-tickets}
	For any job $j \in \cJ$ and machine $i \in \cM$, we have, 
	\begin{itemize}
		\item $\Pr[ N_{ij} = 0 \; | \; \neg B_{ij}]  = 1$. 
		\item $\Pr[ N_{ij} = 0 \; | \; B_{ij} ] = 1- \frac{1- \exp(-x_{ij})}{x_{ij}}$.
	\end{itemize}	
\end{proposition}
\begin{proof}
	The first claim is immediate from the algorithm definition, that is, $N_{ij} = \tilde N_{ij} B_{ij}$. 	The second claim follows since  $\Pr[ N_{ij} > 0 \; | \; B_{ij}] = \Pr[ \tilde N_{ij} > 0] = 	\sum_{k > 0} e^{-\lambda} \frac{\lambda^k}{\lambda k!} = (1 - \exp(-\lambda)) / \lambda$
where $\lambda = x_{ij}$. 	
\end{proof}

\begin{lemma}
	\label{lem:at-least-one-of-two}
	For any two jobs $j \neq j' \in \cJ$, 
	$\Pr[ nil \leftarrow_1 j  \wedge nil \leftarrow_1 j'] \leq 1/ e^2$.
\end{lemma}
\begin{proof}
	The event occurs if and only if no tickets are generated for $j$ and $j'$, i.e., $\sum_{i \in \cM} (N_{ij} + N_{ij'}) = 0$. We will show that for any machine $i$, 
\begin{equation}	
	\label{eqn:1-0}
	\Pr[ N_{ij} + N_{ij'} = 0] \leq e^{-x_{ij}-x_{ij'}}. 
\end{equation}	
	Since the events $\{N_{ij} + N_{ij'} = 0 \}_{i \in \cM}$ are independent, Eqn. (\ref{eqn:1-0}) would imply $\Pr[ \sum_{i \in \cM} (N_{ij} + N_{ij'}) = 0] = \prod_{i \in \cM} \Pr[ N_{ij} + N_{ij'} = 0] \leq \prod_{i \in \cM} \exp( - x_{ij} - x_{ij'}) = \exp( - \sum_{i \in \cM} x_{ij} - \sum_{i \in \cM} x_{ij'}) = 1/e^2$, as desired. 
	
		It now remains to show Eq. (\ref{eqn:1-0}). Towards this end, fix a machine $i \in \cM$. For notational convenience, let $a := x_{ij}$ and $b := x_{ij'}$.  There are two cases we need to consider. If $j$ and $j'$ are not grouped together on machine $i$, the two RVs $N_{ij}$ and $N_{ij'}$ are independent. By Observation~\ref{o:num-tickets}, we know that $N_{ij} \sim \poi(x_{ij} = a)$ and $N_{ij'} \sim \poi(x_{ij'} = b)$. Therefore, $\Pr[ N_{ij} = 0 \wedge N_{ij} = 0] = \Pr[ N_{ij} = 0] \cdot \Pr[N_{ij} = 0] = e^{-a} \cdot e^{-b}$.

	We now consider the other case where $j, j' \in G$ for some $G \in \cG_i$. Depending on which job (or none) is recommended by the group $G$, we consider three disjoint events: $\neg B_{ij} \wedge \neg B_{ij'}$, $B_{ij}$, $B_{ij'}$. 
\begin{align}
	&\Pr[ N_{ij} + N_{ij'} = 0] \nonumber \\
	=& \Pr[ N_{ij} + N_{ij'} = 0 \wedge \neg B_{ij} \wedge \neg \wedge B_{ij'}] \nonumber\\
	&+ \Pr[ N_{ij} + N_{ij'} = 0 \wedge B_{ij} ] + \Pr[ N_{ij} + N_{ij'} = 0 \wedge B_{ij'}] \nonumber \\
	=&\Pr[ \neg B_{ij} \wedge \neg \wedge B_{ij'}]  + \Pr[ N_{ij}  = 0 \wedge B_{ij} ] \nonumber\\
	&+ \Pr[  N_{ij'} = 0 \wedge B_{ij'}] \nonumber \\
	=&(1 - a - b)+ \Pr[ N_{ij}  = 0 \; | \; B_{ij} ] \cdot \Pr[ B_{ij}] \nonumber\\
	&+ \Pr[  N_{ij'} = 0 \; | \;  B_{ij'}] \cdot \Pr [B_{ij'}] \nonumber \\
	= & ( 1- a - b) +  a \cdot (1 - \frac{1- e^{-a}}{a}) + b \cdot  (1 - \frac{1- e^{-b}}{b})  \  \\
	&\mbox{[Proposition~\ref{pro:conditional-num-tickets}]} \nonumber \\%  	= &e^{-a} + e^{-b} -1  \nonumber\\
	= &e^{-a-b}  - (e^{-a} -1)(e^{-b} -1) \nonumber\\
	\leq &e^{-a-b}  = e^{-x_{ij}-x_{ij'}}, \nonumber
\end{align}
\noindent 
where the second equality follows since $\neg B_{ij}$ and $ \neg B_{ij'}$ imply $N_{ij} = B_{ij} \tilde N_{ij} = 0$ and $N_{ij'} = B_{ij'} \tilde N_{ij'} = 0$, respectively. 
\end{proof}

By Observation~\ref{obs:indep-rounds}, we obtain the following corollary. 

\begin{corollary}
	\label{cor:fourth-both-in-second}
	For any two jobs, $j \neq j' \in \cJ$ and machine $i \in \cM$, we have 
	$\Pr[ i \leftarrow_{\geq 2} j  \wedge i \leftarrow_{ \geq 2} j'] \leq (1/ e^2) \Pr[ i \leftarrow j  \wedge i \leftarrow j']$.
\end{corollary}

\smallskip
We now consider the other case when $i \leftarrow_1 j \wedge i \leftarrow_{\geq 2} j'$ or $i \leftarrow_{\geq 2} j \wedge i \leftarrow_1 j'$. This means that one of the two jobs is assigned to machine $i$ and the other is not assigned to any machines in the first iteration.

\begin{lemma}
	\label{lem:pair}
	For any machine $i \in \cM$ and for any two jobs $j \neq j' \in G$ for some $G \in \cG_i$, we have 
	$\Pr[ i \leftarrow_1 j \wedge nil \leftarrow_1 j'] \leq (1- 1/e) (1/e) e^{x_{ij'}} x_{ij}$.
\end{lemma}
\begin{proof}
	Fix any $i$, $j$, $j'$, $G$ satisfying the lemma precondition. Say $i = 1$ wlog. 
	We will show that $\Pr[ 1 \leftarrow_1 j \wedge nil \leftarrow_1 j']$ is maximized when $j$ and $j'$ are not grouped together on any machines except machine $1$. Towards this end, if there is a machine $i' \neq i$, say $i' = 2$, where $j, j' \in G'$ for some $G' \in \cG_2$, we create another grouping by splitting $G'$ into $A$ and $B$ arbitrarily such that $j \in A$ and $j' \in B$. Then we will show that the probability is no smaller under grouping $\cG'$ than it is under $\cG$. Repeating this argument will prove the  the claim.
	
	Let $P_{\cG}$ denote $\Pr[ i \leftarrow_1 j \wedge nil \leftarrow_1 j']$ under grouping $\cG$. Define $P_{\cG'}$ analogously. Our goal is to show $P_{\cG} \leq P_{\cG'}$. We will compare $P_{\cG}$ and $P_{\cG'}$ fixing $N_{1j}, N_{1j'}$, $\tilde N_{2j}$, $\tilde N_{2j'}$, $N_{-2j} := \sum_{i'' \in \cM, i'' \neq 2} N_{i''j}$, and $N_{-2j'} := \sum_{i'' \in \cM, i'' \neq 2} N_{i''j'}$. Note that $nil \leftarrow_1 j'$ only if 	$N_{1j'} = N_{-2j'}  = 0$. Also, $1 \leftarrow_1 j$ only if $N_{1j} > 0$. So, assume $N_{1j} >0$; then we also have $N_{-2j} \geq N_{1j} >0$. 
	
	Given that the aforementioned RVs are fixed, the only factor that potentially makes $P_{\cG}$ and $P_{\cG'}$ different is how $B_{2j}$ and $B_{2j'}$ are related---the two are disjoint under $\cG$ but are independent under $\cG'$. We first observe that $P_{\cG} = P_{\cG'}$ if $\tilde N_{2j'} = 0$. This is because if $\tilde N_{2j'} = 0$, then $N_{2j'} = 0$ no matter what, meaning that $B_{2j'}$ has no effect on  $P_{\cG}$ or $P_{\cG'}$; further, $\Pr[B_{2j}]$ is the same under both $\cG$ and $\cG'$. So, we assume that the fixed $\tilde N_{2j'} > 0$. 
		
	To calculate $P_{\cG}$ and  $P_{\cG'}$, we consider three disjoint events  depending on the recommendation made by $G'$ on machine $2$, namely $B_{2j}$, $B_{2j'}$, and $\neg B_{2j} \wedge \neg B_{2j'}$. If $B_{2j} = 1$, then $N_{2j} = \tilde N_{2j}$. Note that per the above discussion, we can safely assume that $N_{1j'} = N_{-2j'}  = 0$, $N_{1j} >0$, $\tilde N_{2j'} > 0$. 	So, we have $\Pr[ 1 \leftarrow_1 j \wedge nil \leftarrow_1 j' \wedge B_{2j}] = x_{2j} \cdot \frac{N_{1j}}{\tilde N_{2j} + N_{-2j}} $.  If $B_{2j'} = 1$, then $N_{2j'} =  \tilde N_{2j'} >0$, therefore $j' \in \cJ_{\leq 1}$. Finally, if $\neg B_{2j} \wedge \neg B_{2j'}$, then $N_{2j} = N_{2j'} = 0$ ---so we have $\Pr[ 1 \leftarrow_1 j \wedge nil \leftarrow j' \wedge \neg B_{2j} \wedge \neg B_{2j'}] = (1 - x_{2j} - x_{2j'}) \frac{N_{1j}}{N_{-2j}}$. Therefore, we have, 	
\begin{align*}
	P_{\cG} = \frac{N_{1j}}{N_{-2j} + \tilde N_{2j}} a +  \frac{N_{1j}}{N_{-2j}}  (1 - a -b),  
\end{align*}
where we let $a := x_{2j}$, $b := x_{2j'}$ for notational convenience. 

	We now focus on calculating $P_{\cG'}$. Note that under grouping $\cG'$, $B_{2j}$ and $B_{2j'}$ are independent. For $j' \notin \cJ_{\leq 1}$ to happen, it must be the case that $B_{2j'} = 0$ since the fixed $\tilde N_{2j'} >0$. Then, by considering whether $B_{ij}$ or not, we have, 
\begin{align*}
	P_{\cG'} &= (1- b) \Big(\frac{N_{1j}}{N_{-2j} + \tilde N_{2j}} a +  \frac{N_{1j}}{N_{-2j}}  (1 - a) \Big) 	
\end{align*}

Then, we have, 
\begin{align*}
	P_{\cG'} -P_{\cG}  = ab \Big( - \frac{N_{1j}}{N_{-2j} + \tilde N_{2j}} + \frac{N_{1j}}{N_{-2j}}\Big) \geq 0,   
\end{align*}
\noindent since $N_{-2j} \geq N_{1j}$, as desired. 

\smallskip
Therefore, to upper bound $\Pr[ i \leftarrow_1 j \wedge nil \leftarrow_1 j']$, we can safely assume that no group, except $G$ on machine 1, has  both jobs $j$ and $j'$ simultaneously.  This implies that $N_{-1j} := \sum_{i'' \in \cM, i'' \neq 1} N_{i''j}$ and $N_{-1j'} := \sum_{i'' \in \cM, i'' \neq 1} N_{i''j'}$ are independent. Note that $N_{-1j} \sim \poi(1- x_{1j})$ and $N_{-1j'} \sim \poi(1- x_{1j'})$ by Observations~\ref{o:num-tickets} and \ref{o:tickets-independent}.
We now derive, 
\begin{align*}
	\label{eqn:3}
 	 &\Pr[ 1 \leftarrow_1 j  \wedge nil  \leftarrow_1 j']   \nonumber \\
	 	 = &\Pr[ 1 \leftarrow_1 j \wedge B_{1j} \wedge N_{1j'} =  N_{-1j'} = 0] \\ & \mbox{[$1 \leftarrow_1 j$ only if $B_{1j}$]} \nonumber \\
	 	 = &\Pr[ 1 \leftarrow_1 j \wedge B_{1j} \wedge N_{-1j'} = 0] \\ & \mbox{[$N_{1j'} = 0$ if $B_{1j}$, since $j, j' \in G$ on machine 1)]} \nonumber \\
 	 	 = &\Pr[ 1 \leftarrow_1 j \wedge B_{1j}] \cdot \Pr[ N_{-1j'} = 0]   \\
 	 = &\Pr[ 1 \leftarrow_1 j \; | \; B_{1j} ] \cdot \Pr[B_{1j}] \cdot e^{- (1 - x_{1j'})}   & \\
 	 &\mbox{[$N_{-1j'} \sim \poi(1 - x_{1j'})]$} \\
	 = &(1 - 1/e) e^{- (1 - x_{1j'})}  x_{1j} \qquad\mbox{[Lemma~\ref{lem:single}]} 
\end{align*}
The third equation follows since $N_{-1j'}$ is determined by the recommendation of groups on machines other than $1$, which don't include $j$, and $\{\tilde N_{i''j'}\}_{i'' \in \cM, i'' \neq 1}$ -- all these are independent of $B_{1j}$ and $1 \leftarrow_1 j$ since $j$ and $j'$ are not grouped together on any machines except machine $1$.
\end{proof}

\begin{corollary}
	\label{cor:pair}
	For any machine $i \in \cM$ and for any two jobs $j \neq j' \in G$ for some $G \in \cG_i$, we have 
\begin{align*}
	\Pr[ i \leftarrow_1 j \wedge i \leftarrow_{\geq 2} j'] &\leq (1- 1/e) (1/e) e^{x_{ij'}} x_{ij}x_{ij'}; \textnormal{ and} \\
	\Pr[ i \leftarrow_{\geq 2} j \wedge i \leftarrow_{1} j'] &\leq (1- 1/e) (1/e) e^{x_{ij}} x_{ij}x_{ij'}.	
\end{align*}
\end{corollary}
\begin{proof}
	We show the first inequality as follows. 
\begin{align*}
	  &\Pr[ i \leftarrow_1 j \wedge i \leftarrow_{\geq 2} j'] \\
	=&\Pr[ i \leftarrow_1 j \wedge i \leftarrow_{\geq 2} j' \wedge nil \leftarrow_1 j'] \\
	=&\Pr[ i \leftarrow_{\geq 2} j' \; | \; i \leftarrow_1 j \wedge nil \leftarrow_1 j'] \cdot \Pr[ i \leftarrow_1 j \wedge nil \leftarrow_1 j'] \\
	=&\Pr[ i \leftarrow_{\geq 1} j'] \cdot \Pr[ i \leftarrow_1 j \wedge nil \leftarrow_1 j']  \\
	\leq& x_{ij'} (1 - 1/e) e^{- (1 - x_{1j'})}  x_{ij},
\end{align*}
\noindent where the penultimate equality is due to Observation~\ref{obs:indep-rounds}; and the last equality is due to the second property and 
	Lemma~\ref{lem:pair}. The second inequality can be analogously shown using its symmetricity to the first.
\end{proof}

We are now ready to complete the proof of the fourth property. For any two jobs $j \neq j' \in G$ for some $G \in \cG_i$, we know $i \leftarrow_1 j$ and $i \leftarrow_1 j'$ cannot happen simultaneously since $B_{ij}$ and $B_{ij'}$ are disjoint. Therefore, by Corollaries~\ref{cor:fourth-both-in-second} and~\ref{cor:pair}, we have
\begin{align*}
	  &\Pr[ i \leftarrow j \wedge i \leftarrow  j']  \\
  	 \leq &\Pr[ i \leftarrow_1 j \wedge i \leftarrow_{\geq 2} j']\\
  	 &+  \Pr[ i \leftarrow_2 j \wedge i \leftarrow_{\geq 1} j'] + \Pr[ i \leftarrow_{\geq 2} j \wedge i \leftarrow_{\geq 2} j'] \\
  	 \leq &(1- 1/e) (1/e) (e^{x_{ij}} + e^{x_{ij'}})  x_{ij} x_{ij'}\\
  	 &+ (1 / e^2)\Pr[ i \leftarrow j \wedge i \leftarrow  j']
\end{align*}

By rearranging terms, we have 
\begin{align*}
	  &\Pr[ i \leftarrow j \wedge i \leftarrow  j']  \\
  	 &\leq  \frac{(1- 1/e) (1/e)}{1 - 1/e^2}  (e^{x_{ij}} + e^{x_{ij'}})  x_{ij} x_{ij'} \\   	& =  \frac{1}{e+1}  (e^{x_{ij}} + e^{x_{ij'}})  x_{ij} x_{ij'} 
\end{align*}
This completes the proof of the fourth property. 

%% file: 3.sched.tex
\section{Unrelated Machines Scheduling: Rounding Algorithm}
\label{sec:unrelated-algo}
In this section we describe how we round the  optimal  solution to the LP described in Section~\ref{sec:application}. 

As mentioned, we will view the LP solution $\{x_{ijs}\}_{i \in \cM, j \in \cJ, s}$ as a collection of rectangles. 

\begin{definition}
    For every $x_{ijs} >0$, where $i \in \cM, j \in \cJ, s  \geq 0$, there is a rectangle $R_{ijs}$ that starts at time $s$ and ends at time $s+p_{ij}$ of height $x_{ijs}$.
\end{definition}

\begin{definition}
    The height of job $j \in \cJ$ on machine $i \in \cM$ is defined as $x_{ij} := \sum_{s} x_{ijs}$.
\end{definition}

To apply Theorem~\ref{thm:rounding-rephrase}, we need to group jobs on each machine, which will be done stochastically.   For every $j \in \cJ$ and $i \in \cM$ such that $x_{ij} >0$, we choose $\tau_{ij}$ uniformly at random from $(0,p_{ij}]$; and select one rectangle $R_{ijs}$ with probability $\frac{x_{ijs}}{x_{ij}}$ as the \emph{representative} rectangle of job $j$ on machine $i$, which is  denoted as $R_{ij}$.

 We now define good jobs and bad jobs on each machine. Intuitively, a job is good on a machine if it is scheduled a lot on the machine in the LP solution or it starts considerably late compared its size (the starting point of the job's representative rectangle $R_{ij}$ is large compared to its size $p_{ij}$)---we say a job is good because if all jobs are good, then we can obtain a better than 1.5-approximation  by independent rounding. 
 
\begin{definition} 
	Given $R_{ij} = R_{ijs}$, job $j \in \cJ$ is good on machine $i \in \cM$ if $s \geq \frac{1}{10}p_{ij}$ or $x_{ij} \ge \frac{9}{100}$; otherwise, job $j$ is bad on machine $i$. 
\end{definition}

\begin{definition} 
    A rectangle $R_{ijs}$ is good if $s \geq \frac{1}{10} p_{ij}$ or $x_{ij} \geq \frac{9}{100}$, otherwise bad. 
\end{definition}

Note that a job $j$ is good on machine $i$ if its representative rectangle $R_{ij}$ on machine $i$ is good. So, job $j$ being good or bad is stochastic. 

Before grouping jobs, we shift every rectangle to the right---how much a rectangle $R_{ij} = R_{ijs}$ is shifted depends on its start time and how much $j$ is scheduled on the machine: Rectangle $R_{ij} = R_{ijs}$ is shifted to the right by 
\begin{itemize}
    \item $0.34(s+x_{ij}p_{ij})$ if $x_{ij}\ge\frac{9}{100}$
    \item $0.34s$ otherwise. 
\end{itemize}

 We denote the starting point of $R_{ijs}$ after shifting as $\hat{s}$. For each job $j$ on machine $i$ with $R_{ij} = R_{ijs}$ and random offset $\tau_{ij}$, we set $\theta_{ij}=\hat{s}+\tau_{ij}$. We will schedule the jobs assigned to each machine $i$ in increasing order of their $\theta$ values on machine $i$.\\
 
 Now to decide where to assign each job $j$, we use Theorem~\ref{thm:rounding-rephrase}. To apply the theorem we need to group jobs on each machine.
 To define the grouping $\cG_i$ on each machine $i$,  fix a machine $i \in \cM$. We only group bad jobs on machine $i$ -- or equivalently, we create a singleton group for each job that is not grouped together with any other jobs on the machine.
 To group bad jobs on machine $i$, we  use a set of random grid points that are exponentially increasing. Choose a number $\rho$ from $(1/ 10, 1)$ uniformly at random. A point in time is called a grid time if it is of the form $\rho 10^l$ for some integer $l$. Define the $k$-th grid interval, $I_k := (\rho 10^k, \rho 10^{k+1})$. Note that grid times partition the time horizon $(0, \infty)$ into $\{I_k\}_{k \in \mathbb{Z}}$.

We associate   job $j$ with interval $I_k$ on machine $i$ if $\theta_{ij} \in I_k $ and an independent coin  gives a head---the coin gives 
a head with probability $u = 1/2$---and we denote  this event as 
$j\leadsto_i I_k$. 
Now consider all jobs associated with an interval $I_k$ on machine $i$. If their total height is less than 1, i.e., $\sum_{j: j \leadsto_i I_k} x_{ij} \leq 1$, we group them together and add this group to $\cG_i$. 
If two jobs $j$ and $j'$ are grouped together on machine $i$, we will denote the event as $j \stackrel{i}{\sim} j'$; otherwise 
$j \stackrel{i}{\nsim} j'$. 
We now use Theorem~\ref{thm:rounding-rephrase} to decide where to assign each job.

As mentioned above, given that every job is assigned to a machine, we order the jobs assigned to the same machine $i$, in increasing order of their $\theta_{ij}$ value. This completes the description of our randomized rounding. 

%% file: 4.analysis-sched0.tex
\section{Unrelated Machines Scheduling: Analysis}
\label{sec:unrelated-analysis}

This section is devoted to the analysis of the randomized algorithm in Section~\ref{sec:unrelated-algo} with the goal of proving Theorem~\ref{thm:completion-time}.

We first remind the reader of the following fact.

\begin{observation}
\label{obs:secondStep-properties-hold}
All properties of our randomized rounding method (stated in Theorem~\ref{thm:rounding-rephrase}) hold true for any fixed value of 
the random variables $\{\theta_{ij}\}_{i,j}$ and $\rho$.
In particular,  $\Pr[i \leftarrow j] = x_{ij}$ for any $\theta$ values of jobs and $\rho$ value and is independent of how jobs are grouped together on each machine and which rectangles are chosen as jobs' representative rectangles.
\end{observation}

 Let $\tilde C_{\Sj}$ be the completion time of a fixed job $\Sj$ in the schedule returned by the rounding algorithm. Using conditional expectation and the law of total expectation, we have
\begin{align*}
E[\tilde C_{\Sj}]&= 
\sum_{i \in \cM} \sum_{\Ss} \Pr[i \leftarrow \Sj] \Pr[R_{i\Sj}=R_{i\Sj\Ss} \;|\; i \leftarrow \Sj ]~E[\tilde{C}_{\Sj}\;|\; i \leftarrow \Sj, R_{i\Sj}=R_{i\Sj\Ss}] \\
&=\sum_{i \in \cM} \sum_{\Ss} x_{i\Sj} \Pr[R_{i\Sj}=R_{i\Sj\Ss} \;|\; i \leftarrow \Sj ]~E[\tilde{C}_{\Sj}\;|\; i \leftarrow \Sj, R_{i\Sj}=R_{i\Sj\Ss}] 
\end{align*}

As stated in 
Observation~ \ref{obs:secondStep-properties-hold}, the event $i \leftarrow \Sj$
is independent of $R_{i\Sj}=R_{i\Sj\Ss}$, which happens with  probability $\frac{x_{i\Sj\Ss}}{x_{i\Sj}}$. Thus,  we can simplify the above as follows.
\begin{equation}
\label{eq:main-expectation}
 \E[\tilde C_{\Sj}]=\sum_{i \in M} \sum_{\Ss} x_{i\Sj\Ss} \E[\tilde{C}_{\Sj}\;|\;i \leftarrow \Sj, R_{i\Sj}=R_{i\Sj\Ss}] 
\end{equation} 
From now on, we fix machine $i$. Our key lemma in the analysis is the following. 
\begin{lemma}
\label{lem:main-sched}
  For each job $\Sj$, we have
$$\sum_{\Ss} x_{i\Sj\Ss} E[\tilde{C}_{\Sj}\;|\; i \leftarrow \Sj, R_{i\Sj}=R_{i\Sj\Ss}]  \le 1.488 \sum_{\Ss} x_{i\Sj\Ss}(\Ss+p_{i\Sj}) $$

\end{lemma}
Plugging this inequality into Eqn.(\ref{eq:main-expectation}) would immediately imply Theorem~\ref{thm:completion-time}: 
$$\E[\tilde C_{\Sj}] \le 1.488 \sum_i \sum_{\Ss} x_{i\Sj\Ss}(\Ss+p_{i\Sj}),$$
\noindent as summing $\E[\tilde C_{\Sj}]$ over all jobs multiplied by their weight gives an upper bound of our algorithm's expected objective by $1.488$ times the LP optimum objective.

\medskip

The rest of this section is devoted to proving Lemma~\ref{lem:main-sched}.  Following observations will be useful for our analysis.
\begin{observation}
\label{obs:expectation-block}
For any job $\Sj$ and any rectangle $R_{i\Sj\Ss}$ we have
\begin{equation}
\label{eq:expectation-block}
E[\tilde{C}_{\Sj}\;|\; i \leftarrow \Sj, R_{ij^*} = R_{i\Sj\Ss}] = \int_{0}^{p_{i\Sj}} \frac{1}{p_{i\Sj}}E[\tilde{C}_{\Sj}\;|\; i \leftarrow \Sj, R_{ij^*} = R_{i\Sj\Ss},\tau_{i\Sj}=\tau] \; d\tau 
\end{equation}
\end{observation}
\begin{proof}
Recall that $\tau_{i\Sj}$ is chosen from $(0,p_{i\Sj}]$ uniformly at random. Using conditional expectation and the law of total expectation, we have,
\begin{align*}
&E[\tilde{C}_{\Sj}\;|\; i \leftarrow \Sj,R_{ij^*} = R_{i\Sj\Ss}] \\=& \int_{0}^{p_{i\Sj}} \Pr[\tau_{i\Sj}=\tau \;|\; i \leftarrow \Sj, R_{ij^*} = R_{i\Sj\Ss}]~E[\tilde{C}_{\Sj}\;|\; i \leftarrow \Sj, R_{ij^*} = R_{i\Sj\Ss},\tau_{i\Sj}=\tau] \; d\tau 
\end{align*}
Note that $\tau_{i\Sj}$ is independent of $R_{i\Sj}=R_{i\Sj\Ss}$, and 
$i \leftarrow j^*$ is independent of these two events as observed in Observation~\ref{obs:secondStep-properties-hold}. Thus, we have
  $\Pr[\tau_{i\Sj} \in [\tau, \tau + d\tau) \;|\; i \leftarrow \Sj, R_{i\Sj}=R_{i\Sj\Ss}] =\Pr[\tau_{i\Sj}\in [\tau, \tau + d\tau)]= \frac{d \tau}{p_{i\Sj}} $. This completes the proof. 
\end{proof}

For the sake of analysis we define notation $\hat{L}_{ijs}(\theta)$ for any rectangle $R_{ijs}$ and ant time point $\theta$ to be the length of part of the rectangle $R_{ijs}$ that appears before time $\theta$ \emph{after} shifting. Note that $\hat{L}_{ijs}(\theta)=\min\{\theta,\hat{s}+p_{ij}\}-\hat{s} $ if $\hat{s}<\theta$, otherwise $\hat{L}_{ijs}(\theta)=0$. 

The next observation measures the probability that a job $j$ will have a smaller $\theta$ value than the fixed job $\Sj$ on machine $i$. If this event occurs, and further, both jobs are assigned to the machine $i$, job $j$ will delay job $\Sj$. 
\begin{observation}
\label{obs:event_e}
For any two jobs $\Sj$ and $j$ on machine $i$,  any fixed value $\theta_{i\Sj}$ and any $R_{ijs}$, we have 
 $\Pr[\theta_{ij}\le \theta_{i\Sj}\; | \; R_{ij}=R_{ijs} ]=\frac{\hat{L}_{ijs}(\theta_{i\Sj})}{p_{ij}}$.
\end{observation}
\begin{proof}
Since the representative rectangle of job $j$ on machine $i$ is fixed, the event $\theta_{ij}=\hat{s}+\tau_{ij}\le \theta_{i\Sj}$ happens iff $\tau_{ij}\in (0,\max\{0, \min\{\theta_{i\Sj},\hat{s}+p_{ij}\}-\hat{s}\}]=(0,\hat{L}_{ijs}(\theta_{i\Sj})]$. Since $\tau_{ij}$ is chosen from $(0,p_{ij}]$ uniformly at random, the event $\tau_{ij} \in (0,\hat{L}_{ijs}(\theta_{i\Sj})]$ occurs with probability $\frac{\hat{L}_{ijs}(\theta_{i\Sj})}{p_{ij}}$. 
\end{proof}

\subsection{Proof of the Lemma~\ref{lem:main-sched}}

We now get back to proving our main lemma, Lemma~\ref{lem:main-sched}. Towards that end, we first express $j^*$'s expected completion times in terms of the quantities we defined. The following lemma breaks down $j^*$'s expected completion time conditioned on the fixed job $j^*$ being assigned to machine $i$, the rectangle $R_{i\Sj\Ss}$ being selected as $\Sj$'s representative on machine $i$ and the value of $\theta_{i, j^*}$ being fixed. Note that $x_{ijs}\cdot \hat{L}_{ijs}(\theta_{i\Sj})$  is exactly the volume of rectangle $R_{ijs}$ appearing before time $\theta_{i\Sj}$ after shifting and $\sum_{j ,s}x_{ijs}\cdot \hat{L}_{ijs}(\theta_{i\Sj}) - \sum_{s}x_{i\Sj s}\cdot \hat{L}_{i\Sj s}(\theta_{i\Sj})$ is total volume of all rectangles excluding those of job $\Sj$ appearing before time $\theta_{i\Sj}$ after shifting---this is exactly how much other jobs would delay job $j^*$ if we used the standard independent rounding. The gain comes from representative rectangles that are grouped together with the fixed $R_{i\Sj\Ss}$, in which jobs of such rectangles are less likely to be assigned to the same machine $i$ due to the strong negative correlation property. 

For brevity, we  may shorten $R_{ij} = R_{ijs}$ simply as $R_{ijs}$ particularly when it is stated in the condition of a probability or expectation. Throughout this paper, we set $\eta := \frac{2\exp(0.09)}{e+1} < 0.589$, which comes from the following: Recall the fourth property of Theorem~\ref{thm:rounding-rephrase}.  For any two distinct bad jobs $j$ and $j'$, as $x_{ij},x_{ij'} <9/100$, we have $\Pr[i \leftarrow j \wedge i \leftarrow j' \;|\; j \stackrel{i}{\sim} j'] \le \frac{2\exp(0.09)}{e+1} x_{ij}x_{ij'} <  0.589 x_{ij}x_{ij'}$.

\begin{lemma}
\label{lem:Expectation-block-tau}
For any job $\Sj$, representative rectangle $R_{i\Sj\Ss}$ and any fixed $\tau_{i\Sj} \in (0, p_{i\Sj}]$, we have,
\begin{align*}
 	& E\Big[\tilde C_{j^*} \; \Big|  \;i \leftarrow \Sj , R_{i\Sj\Ss},\tau_{i\Sj}\Big]  - p_{i\Sj} \\
&= \sum_{j ,s}x_{ijs}\cdot \hat{L}_{ijs}(\theta_{i\Sj}) - \sum_{s}x_{i\Sj s}\cdot \hat{L}_{i\Sj s}(\theta_{i\Sj})- (1 - \eta) \sum_{j \neq \Sj,s}x_{ijs}\cdot p_{ij} \cdot \Pr[e_{ij}, j \stackrel{i}{\sim} \Sj\;|\; R_{ijs},R_{i\Sj\Ss},\tau_{i\Sj}]
\end{align*}
where $e_{ij}$ denotes the event that $\theta_{ij}\le \theta_{i\Sj}$.

\end{lemma}
\begin{proof}
For brevity, we omit the fixed $\tau_{i\Sj}$ and representative rectangle $R_{i\Sj\Ss}$ from the condition. Then, we have,
\begin{align*}
&E\Big[\tilde C_{\Sj} \;\Big|  \;i \leftarrow \Sj \Big]  - p_{i\Sj} \nonumber\\
=&\sum_{ j \neq j^*, e_{ij} } \Pr\Big[i \leftarrow j \; \Big|\; i \leftarrow \Sj \Big] \cdot p_{ij} \\
=&\;\;\;\;\;\; \sum_{ j \neq j^*,e_{ij}, j \stackrel{i}{\nsim} \Sj  } \Pr\Big[i \leftarrow j \; \Big|\; i \leftarrow \Sj \Big] \cdot p_{ij} 
 +\sum_{ j \neq j^*, e_{ij}, j \stackrel{i}{\sim} \Sj  } \Pr\Big[i \leftarrow j \; \Big|\; i \leftarrow \Sj \Big] \cdot p_{ij} \\
\leq&\;\;\;\;\;\; \sum_{ j \neq j^*, e_{ij}, j \stackrel{i}{\nsim} \Sj  } x_{ij} \cdot p_{ij} 
   +\eta \sum_{ j \neq j^*, e_{ij}, j \stackrel{i}{\sim} \Sj  } x_{ij} \cdot p_{ij} \\
=&\;\;\;\;\;\; \sum_{ j \neq j^*, e_{ij}  } x_{ij} \cdot p_{ij}   -(1-\eta) \sum_{ j \neq j^*, e_{ij}, j \stackrel{i}{\sim} \Sj  } x_{ij} \cdot p_{ij} 
\end{align*}
where the inequality follows from negative and strong negative correlation properties of our randomized rounding algorithm stated in Theorem~\ref{thm:rounding-rephrase}.

Note that fixing the representative rectangle $R_{i\Sj\Ss}$ and the value of $\tau_{ij^*}$ fixes the value of $\theta_{i\Sj}$. Then, thanks to Observation~\ref{obs:event_e}, the first term becomes 
\begin{align*}
&\sum_{j \neq \Sj,s} x_{ij} \cdot p_{ij} \cdot \Pr[R_{ij} = R_{ijs},e_{ij}]=\sum_{j \neq \Sj,s} x_{ij} \cdot p_{ij} \cdot \Pr[e_{ij}\;|\; R_{ijs}] \cdot \Pr[R_{ij} = R_{ijs}]\\
&=\sum_{j \neq \Sj,s}x_{ij}\cdot\frac{x_{ijs}}{x_{ij}}\cdot p_{ij} \cdot\frac{\hat{L}_{ijs}(\theta_{i\Sj})}{p_{ij}}\quad \quad \mbox{[Observation~\ref{obs:event_e}]}
\\
&=\sum_{j \neq \Sj,s}x_{ijs}\cdot \hat{L}_{ijs}(\theta_{i\Sj})\\
&=\sum_{j ,s}x_{ijs}\cdot \hat{L}_{ijs}(\theta_{i\Sj}) - \sum_{s}x_{i\Sj s}\cdot \hat{L}_{i\Sj s}(\theta_{i\Sj})
\end{align*}
The second term becomes
\begin{align*}
&(1 - \eta) \sum_{j \neq \Sj,s} x_{ij} \cdot p_{ij} \cdot \Pr[R_{ijs},e_{ij}, j \stackrel{i}{\sim} \Sj]\\
&=(1 - \eta) \sum_{j \neq \Sj,s} x_{ij} \cdot p_{ij} \cdot \Pr[e_{ij}, j \stackrel{i}{\sim} \Sj\;|\; R_{ijs}] \cdot \Pr[R_{ijs}]\\
&=(1 - \eta) \sum_{j \neq \Sj,s}x_{ijs}\cdot p_{ij} \cdot \Pr[e_{ij}, j \stackrel{i}{\sim} \Sj\;|\; R_{ijs}]
\end{align*}
Thus, bringing back the omitted conditions $\tau_{ij^*}$ and $R_{i\Sj\Ss}$ to the equation, we have the lemma. 
\end{proof}
\begin{lemma}
\label{lem:rect-upper-ind}
For each job $\Sj$, and any representative rectangle $R_{i\Sj\Ss}$,  
we have 
\begin{equation}
\label{eq:rect-upper-ind}
E[\tilde C_{j^*} \;|  \;i \leftarrow j^*, R_{i\Sj\Ss}] \leq \hat{s}^*+1.5p_{i\Sj}-\int_{0}^{p_{i\Sj}}  \frac{1}{p_{i\Sj}}\cdot \sum_{s}x_{i\Sj s}\cdot \hat{L}_{i\Sj s}(\hat{s}^*+\tau)\; d\tau
\end{equation}
\end{lemma}
\begin{proof}
The trivial lower bound for the second negative term of the lemma~\ref{lem:Expectation-block-tau} is zero. Thus 
$$E[\tilde C_{j^*} \; |  \;i \leftarrow \Sj , R_{i\Sj\Ss},\tau_{i\Sj}]  \leq p_{i\Sj}+
\sum_{j ,s}x_{ijs}\cdot \hat{L}_{ijs}(\theta_{i\Sj}) - \sum_{s}x_{i\Sj s}\cdot \hat{L}_{i\Sj s}(\theta_{i\Sj}) $$

Note that the total volume of all rectangles appearing before time $\theta$ is at most $\theta$. Further, $\sum_{j ,s}x_{ijs}\cdot \hat{L}_{ijs}(\theta)$ is total volume of all rectangles appearing before time $\theta$ after shifting. Since we shift rectangles to the right side, we have $\sum_{j ,s}x_{ijs}\cdot \hat{L}_{ijs}(\theta)\le \theta$. Thus  we have 
\begin{align*}
E[\tilde C_{j^*} \; |  \;i \leftarrow \Sj , R_{i\Sj\Ss},\tau_{i\Sj}]  &= p_{i\Sj}+
\theta_{i\Sj}- \sum_{s}x_{i\Sj s}\cdot \hat{L}_{i\Sj s}(\theta_{i\Sj})\\
&=p_{i\Sj}+\hat{s}^*+\tau_{i\Sj}-\sum_{s}x_{i\Sj s}\cdot \hat{L}_{i\Sj s}(\theta_{i\Sj})
\end{align*}
Thus, taking the integral on the value of $\tau_{ij^*}$, we have, 
\begin{align*}
E[\tilde C_{j^*} \; |  \;i \leftarrow \Sj , R_{i\Sj\Ss}] &\le p_{i\Sj}+ \hat{s}^*+\int_{0}^{p_{i\Sj}} \tau\cdot \frac{1}{p_{i\Sj}} \; d\tau- \int_{0}^{p_{i\Sj}} \big(\sum_{s}x_{i\Sj s}\cdot \hat{L}_{i\Sj s}(\hat{s}^*+\tau)\big)\cdot \frac{1}{p_{i\Sj}} \; d\tau\\
&\le 1.5 \cdot p_{i\Sj}+ \hat{s}^*- \int_{0}^{p_{i\Sj}} \big(\sum_{s}x_{i\Sj s}\cdot \hat{L}_{i\Sj s}(\hat{s}^*+\tau)\big)\cdot \frac{1}{p_{i\Sj}} \; d\tau
\end{align*}
\noindent as desired. 
\end{proof}

\begin{lemma}
	\label{lem:all-rects-upper-ind}
For each job $j$ we have
$$\sum_{\Ss} x_{i\Sj\Ss}E[\tilde C_{j^*} \; |  \;i \leftarrow \Sj , R_{i\Sj\Ss}] \le (1.5-\frac{x_{i\Sj}}{2})x_{i\Sj}p_{i\Sj} \;+\; \sum_{\Ss} x_{i\Sj\Ss} \hat{s}^* $$

\end{lemma}
\begin{proof}
Applying  Lemma~\ref{lem:rect-upper-ind} to all rectangles $\{R_{i\Sj\Ss}\}_{\Ss}$ of job $\Sj$ on machine $i$, we can say
\begin{equation}
\label{eq:all-rects-upper-ind}
\sum_{\Ss} x_{i\Sj\Ss}E[\tilde C_{j^*} \; |  \;i \leftarrow \Sj , R_{i\Sj\Ss}] \le \sum_{\Ss} x_{i\Sj\Ss} \Big( \hat{s}^*+1.5p_{i\Sj}-\int_{0}^{p_{i\Sj}}  \big(\sum_{s}x_{i\Sj s}\cdot \hat{L}_{i\Sj s}(\hat{s}^*+\tau)\big)\cdot \frac{1}{p_{i\Sj}} \; d\tau \Big)
\end{equation}
Consider negative term here.
\begin{align}
&\sum_{\Ss} x_{i\Sj\Ss}\int_{0}^{p_{i\Sj}}  \big(\sum_{s}x_{i\Sj s}\cdot \hat{L}_{i\Sj s}(\hat{s}^*+\tau)\big)\cdot \frac{1}{p_{i\Sj}} \; d\tau \nonumber\\
&=\sum_{\{\Ss,s\}:s\neq\Ss} x_{i\Sj\Ss}\cdot x_{i\Sj s} \int_{0}^{p_{i\Sj}}  \Big( \hat{L}_{i\Sj s}(\hat{s}^*+\tau) + \hat{L}_{i\Sj \Ss}(\hat{s}+\tau)\Big)\cdot \frac{1}{p_{i\Sj}} \; d\tau \nonumber\\
&+\sum_{\Ss}x_{i\Sj\Ss}^2 \int_{0}^{p_{i\Sj}} \hat{L}_{i\Sj \Ss}(\hat{s}^*+\tau)\cdot \frac{1}{p_{i\Sj}} \; d\tau \label{main-eq}
\end{align}

We can show that for any two different $\Ss$ and $s$ we have
\begin{equation}
\label{eq:mutalDelay}
\int_{0}^{p_{i\Sj}}  \Big( \hat{L}_{i\Sj s}(\hat{s}^*+\tau) + \hat{L}_{i\Sj \Ss}(\hat{s}+\tau)\Big)\cdot \frac{1}{p_{i\Sj}} \; d\tau = p_{i\Sj}
\end{equation}
W.l.o.g we can assume $\hat{s}^* <\hat{s}$. Consider the case that two rectangles $R_{i\Sj\Ss}$ and $R_{i\Sj s}$ do not overlap after shifting. In this case, for any $\tau$, $\hat{L}_{i\Sj \Ss}(\hat{s}+\tau)=p_{i\Sj}$. Thus the equality (\ref{eq:mutalDelay}) holds in this case. Now consider the case that two rectangles $R_{i\Sj\Ss}$ and $R_{i\Sj s}$ overlap after shifting. In this case when $\tau \in (0,\hat{ s}-\hat{s}^*]$, $\hat{L}_{i\Sj s}(\hat{s}^*+\tau)=0$. When $\tau \in (\hat{ s}-\hat{s}^*,p_{i\Sj}]$, $\hat{L}_{i\Sj s}(\hat{s}^*+\tau)=\hat{s}^*+\tau-\hat{s}$. For any $\tau$, $\hat{L}_{i\Sj \Ss}(\hat{s}+\tau)=\min\{\hat{s}+\tau-\hat{s}^*,p_{i\Sj}\}$. Thus we have  
\begin{align*}
&\int_{0}^{p_{i\Sj}}  \big( \hat{L}_{i\Sj s}(\hat{s}^*+\tau) + \hat{L}_{i\Sj \Ss}(\hat{s}+\tau)\big)\cdot \frac{1}{p_{i\Sj}} \; d\tau\\
&=\int_{\hat{ s}-\hat{s}^*}^{p_{i\Sj}}(\hat{s}^*+\tau-\hat{s})\cdot\frac{1}{p_{i\Sj}} \; d\tau\\
&+\int_{0}^{\Ss+p_{i\Sj}-s}(s+\tau-\Ss)\cdot\frac{1}{p_{i\Sj}} \; d\tau\\
&+\int_{\Ss+p_{i\Sj}-s}^{p_{i\Sj}}(p_{i\Sj})\cdot\frac{1}{p_{i\Sj}} \; d\tau\\
&=p_{i\Sj}
\end{align*}
Since $\hat{L}_{i\Sj \Ss}(\hat{s}^*+\tau)=\tau$ we  have 
\begin{equation}
\label{eq:selfDelay}
\int_{0}^{p_{i\Sj}} \hat{L}_{i\Sj \Ss}(\hat{s}^*+\tau)\big)\cdot \frac{1}{p_{i\Sj}} \; d\tau=p_{i\Sj}/2
\end{equation}
Applying Eqn.(\ref{eq:mutalDelay}) and Eqn.(\ref{eq:selfDelay}) in Eqn.(\ref{main-eq}),the negative term of Eqn.(\ref{eq:all-rects-upper-ind}) becomes:
\begin{align*}
&\sum_{\{\Ss,s\}:s\neq\Ss} x_{i\Sj\Ss}\cdot x_{i\Sj s} \cdot p_{i\Sj} +\sum_{\Ss}x_{i\Sj\Ss}^2 \cdot \frac{p_{i\Sj}}{2}\\
&= \frac{p_{i\Sj}}{2}\Big(\sum_{\Ss}x_{i\Sj\Ss}\Big)^2\\
&= \frac{p_{i\Sj\cdot} x_{i\Sj}^2}{2}
\end{align*}

Plugging this equality in Eqn.(\ref{eq:all-rects-upper-ind}) gives the lemma.
\end{proof}

\medskip
We now consider three cases in the following to prove Lemma~\ref{lem:main-sched}.

\subsubsection{When job $\Sj$ is good on machine $i$ because $x_{i\Sj}>\frac{9}{100}$}

The following corollary follows from Lemma~\ref{lem:all-rects-upper-ind}; the proof is deferred to the full version of this paper. 

\begin{corollary}
When job $\Sj$ is good on $i$ because $x_{i\Sj}>\frac{9}{100}$ we have
$$\sum_{\Ss} x_{i\Sj\Ss}E[\tilde C_{j^*} \; |  \;i \leftarrow \Sj , R_{i\Sj\Ss}] \le
 1.486 \sum_{\Ss}  x_{i\Sj\Ss}(\Ss+p_{i\Sj})$$
\end{corollary}
\begin{proof}
\begin{align*}
\sum_{\Ss} x_{i\Sj\Ss}E[\tilde C_{j^*} \; |  \;i \leftarrow \Sj , R_{i\Sj\Ss}] &\le (1.5-\frac{x_{i\Sj}}{2})x_{i\Sj}p_{i\Sj} \;+\; \sum_{\Ss} x_{i\Sj\Ss}\hat{s}^* \\
& = (1.5-\frac{x_{i\Sj}}{2})x_{i\Sj}p_{i\Sj} \;+\; \sum_{\Ss} x_{i\Sj\Ss}(1.34\Ss+0.34x_{i\Sj}p_{i\Sj})\\
&=(1.5-0.16 x_{i\Sj})x_{i\Sj}p_{i\Sj} \;+\; \sum_{\Ss} 1.34 \cdot x_{i\Sj\Ss}\cdot \Ss\\
&\le (1.5-0.16\times 0.09) \cdot x_{i\Sj}\cdot p_{i\Sj} \;+\; 1.34\sum_{\Ss}  x_{i\Sj\Ss}\cdot \Ss \\
&< 1.486 \sum_{\Ss}  x_{i\Sj\Ss}(\Ss+p_{i\Sj})
\end{align*}
The second equality come from the fact that we shift $R_{ij^*s^*}$ to the right side by $0.34(\Ss+x_{i\Sj}p_{i\Sj})$.
\end{proof}

This gives Lemma~\ref{lem:main-sched} for each good job $\Sj$ with $x_{i\Sj}\ge\frac{9}{100}$. Now we focus on proving Lemma~\ref{lem:main-sched} for each job $\Sj$ with $x_{i\Sj}<\frac{9}{100}$ in the following.

\subsubsection{When job $\Sj$ with $x_{i\Sj}< \frac{9}{100}$ is good on machine $i$ because its representative rectangle $R_{i\Sj\Ss}$ is good on $i$}

The goal of this subsection is to show Lemma~\ref{lem:main-sched} for each job $\Sj$ with $x_{i\Sj}<\frac{9}{100}$ and any good representative rectangle $R_{i\Sj\Ss}$. Towards this end, it suffices to show the following.

\begin{align*}
E[\tilde C_{j^*} \;|  \; i \leftarrow \Sj, R_{i\Sj\Ss}] &\le \hat{s}^*+1.5p_{i\Sj}=1.34\Ss+1.5p_{i\Sj}\\
&=1.486\Ss-0.146\Ss+1.5p_{i\Sj}\\
& \le 1.486 \Ss-0.0146p_{i\Sj}+1.5p_{i\Sj}\\
&\le 1.486(\Ss+p_{i\Sj})
\end{align*}
The first inequality is due to Lemma~\ref{lem:rect-upper-ind}.
The first equality comes from this fact as $x_{i\Sj}<\frac{9}{100}$, we shift rectangle $R_{i\Sj\Ss}$ to the right side by $0.34\Ss$, and the inequality follows from $10\Ss \ge p_{i\Sj}$ for any good rectangle $R_{i\Sj\Ss}$. 

Thus, we have shown Lemma~\ref{lem:main-sched} for this case. 

%% file: 4.analysis-sched1.tex
\subsubsection{When job $\Sj$  is bad on machine $i$}
In this subsection we consider the case that $x_{i\Sj}<\frac{9}{100}$ and the representative rectangle of job $\Sj$ is bad on $i$. To show Lemma~\ref{lem:main-sched} for this case, we will show,
\begin{equation}
\label{eq:badjob}
 \E\Big[\tilde C_{j^*} \; \Big|  \;i \leftarrow \Sj , R_{i\Sj\Ss},\tau_{i\Sj}\Big] \le p_{i\Sj}+\hat{s}^*+0.976\tau_{i\Sj}
\end{equation}
Then, by taking the integral on the value of $\tau_{ij^*}$, we have, 
\begin{align*}
E[\tilde C_{j^*} \; |  \;i \leftarrow \Sj , R_{i\Sj\Ss}] &\le p_{i\Sj}+ \hat{s}^*+\int_{0}^{p_{i\Sj}} 0.976\tau\cdot \frac{1}{p_{i\Sj}} \; d\tau=\hat{s}^*+1.488p_{i\Sj}
\end{align*}

Since $\hat{s}^*=1.34\Ss$, this means we have $\E[\tilde C_{j^*} \; |  \;i \leftarrow \Sj , R_{i\Sj\Ss}] \le 1.488(\Ss+p_{i\Sj})$, as desired.

\smallskip
To prove Eqn.(\ref{eq:badjob}), we will reformat the time-indexed LP solution $\{x_{ijs}\}_{ijs}$, so that we have a linear combination of non-overlapping rectangles on each machine. This view will make our analysis easier. Formally, we will define a collection $\cF_i$ of subsets of rectangles on each machine $i$ that satisfies the following properties.

\begin{enumerate}
\item $\sum_{f \in \cF_i} z_{if} \leq 1$
\item A configuration $f \in \cF_i$ is  a set of disjoint rectangles; that is, for any two  distinct rectangles $R_{ijs}, R_{ij's'}\in f$, $(s,s+p_{ij}]$ and $(s',s'+p_{ij'}]$ are disjoint. 
\item For each rectangle $R_{ijs}$, we have $\sum_{f \ni R_{ijs}}z_{if}=x_{ijs}$.
\end{enumerate}

An easy way to obtain this linear combination decomposition to replace each rectangle with sufficiently many copies of the same height preserving its total height. Then, it is not difficult to see that we can group rectangles as desired. We note that we can directly obtain this type of solution by solving a configuration LP \cite{sviridenko2013approximating}.\footnote{The configuration LP solution is more structured than the above decomposition in that each configuration has at most one rectangle of each job. However, we do not need such a strong property for our rounding.}

\medskip
Now, we restate Lemma~\ref{lem:Expectation-block-tau} using the above configuration view. As we will focus on a fixed machine $i$, we may omit $\cF_i$. Due to the space constraints, we defer the proof of the following corollary to the full version of this paper.

\begin{corollary}
\label{cor:Expectation-block-tau}
For any job $\Sj$, representative rectangle $R_{i\Sj\Ss}$ and any fixed $\tau_{i\Sj} \in (0, p_{i\Sj}]$, we have,
\begin{align*}
 	& E\Big[\tilde C_{j^*} \; \Big|  \;i \leftarrow \Sj , R_{i\Sj\Ss},\tau_{i\Sj}\Big]  - p_{i\Sj} \\
=&\sum_fz_{if}\Big( \hat{L}_{f}(\theta_{i\Sj})- \sum_{R_{i\Sj s}\in f}\hat{L}_{i\Sj s}(\theta_{i\Sj})\Big)\\
&-(1-\eta)\sum_f\sum_{R_{ijs}\in f: j \neq \Sj}z_{if}\cdot p_{ij}  \Pr[e_{ij}, j \stackrel{i}{\sim} \Sj\;|\; R_{ijs},R_{i\Sj\Ss},\tau_{i\Sj}]
\end{align*}
\noindent 
where $\hat{L}_{f}(\theta_{i\Sj})$ is the total length of all rectangles in $f$ appearing up to  time $\theta_{i\Sj}$ after shifting.
\end{corollary}
\begin{proof}
\begin{align*}
 	& E\Big[\tilde C_{j^*} \; \Big|  \;i \leftarrow \Sj , R_{i\Sj\Ss},\tau_{i\Sj}\Big]  - p_{i\Sj} \\
=& \sum_{j ,s}x_{ijs}\cdot \hat{L}_{ijs}(\theta_{i\Sj}) - \sum_{s}x_{i\Sj s}\cdot \hat{L}_{i\Sj s}(\theta_{i\Sj})- (1 - \eta) \sum_{j \neq \Sj,s}x_{ijs}\cdot p_{ij} \cdot \Pr[e_{ij}, j \stackrel{i}{\sim} \Sj\;|\; R_{ijs},R_{i\Sj\Ss},\tau_{i\Sj}]\\
=&\sum_{f}\sum_{R_{ijs}\in f}z_{if}\hat{L}_{ijs}(\theta_{i\Sj}) -\sum_f\sum_{R_{i\Sj s}\in f}z_{if}\hat{L}_{i\Sj s}(\theta_{i\Sj})\\
&-(1-\eta)\sum_f\sum_{R_{ijs}\in f: j \neq \Sj}z_{if}\cdot p_{ij}  \Pr[e_{ij}, j \stackrel{i}{\sim} \Sj\;|\; R_{ijs},R_{i\Sj\Ss},\tau_{i\Sj}]\\
=&\sum_{f}z_{if}\hat{L}_{f}(\theta_{i\Sj}) -\sum_f\sum_{R_{i\Sj s}\in f}z_{if}\hat{L}_{i \Sj s}(\theta_{i\Sj})\\
&-(1-\eta)\sum_f\sum_{R_{ijs}\in f: j \neq \Sj}z_{if}\cdot p_{ij}  \Pr[e_{ij}, j \stackrel{i}{\sim} \Sj\;|\; R_{ijs},R_{i\Sj\Ss},\tau_{i\Sj}]\\
=&\sum_fz_{if}\Big( \hat{L}_{f}(\theta_{i\Sj})- \sum_{R_{i\Sj s}\in f}\hat{L}_{i\Sj s}(\theta_{i\Sj})\Big)\\
&-(1-\eta)\sum_f\sum_{R_{ijs}\in f: j \neq \Sj}z_{if}\cdot p_{ij}  \Pr[e_{ij}, j \stackrel{i}{\sim} \Sj\;|\; R_{ijs},R_{i\Sj\Ss},\tau_{i\Sj}]
\end{align*}

\end{proof}

We define $D_f$ as the contribution of each configuration $f$ to $\E\big[\tilde C_{j^*} \; \Big|  \;i \leftarrow \Sj , R_{i\Sj\Ss},\tau_{i\Sj}\big] $ as following.
$$
D_f=z_{if}\Big( \hat{L}_{f}(\theta_{i\Sj})- \sum_{R_{i\Sj s}\in f}\hat{L}_{i\Sj s}(\theta_{i\Sj})-(1-\eta)\sum_{R_{ijs}\in f:j\neq \Sj} p_{ij} \cdot  \Pr[e_{ij}, j \stackrel{i}{\sim} \Sj\;|\; R_{ijs},R_{i\Sj\Ss},\tau_{i\Sj}]\Big)
$$
With this definition, we can say,
\begin{equation}
\label{eq:expect-based-conf-contribution}
E\Big[\tilde C_{j^*} \; \Big|  \;i \leftarrow \Sj , R_{i\Sj\Ss},\tau_{i\Sj}\Big]=p_{i\Sj}+\sum_f D_f
\end{equation}
Since $\sum_f z_{if}=1$ and $0.976 \theta_{i\Sj} \le \hat{s}^*+0.976\tau_{i\Sj} $, if we show that 
\begin{equation}
    \label{eqn:conf-ultimate}
 D_f \le z_{if}(0.976\theta_{i\Sj}) \text{ for all configurations } f,
\end{equation}
then we can prove Eqn. (\ref{eq:badjob}), and consequently we can get Lemma~\ref{lem:main-sched} for a bad job $\Sj$.\\

\medskip
Henceforth, the goal of our analysis is to show Eqn.~(\ref{eqn:conf-ultimate}). As mentioned there is a random variable $\rho$ in our algorithm that affects grid points,  and therefore, affects the grouping of bad jobs on each machine. Thus, for each $f$, $R_{ijs}, R_{ij^*s} \in f$, where $j \neq j^*$, the probability $ \Pr[e_{ij}, j \stackrel{i}{\sim} \Sj\;|\; R_{ijs},R_{i\Sj\Ss},\tau_{i\Sj}]$ in $D_f$ depends on the value of $\rho$.  As $\rho$ is sampled from  $(\frac{1}{10},1]$ uniformly at random, applying marginal probability and then Bayes' rule, we have 
\begin{equation}
\label{eq:conf-cont-depend-rho}
D_f=z_{if}\Big( \hat{L}_{f}(\theta_{i\Sj})- \sum_{R_{i\Sj s}\in f}\hat{L}_{i\Sj s}(\theta_{i\Sj})-(1-\eta)\sum_{R_{ijs}\in f:j \neq \Sj} p_{ij} \cdot \int_{\frac{1}{10}}^1 \frac{10}{9} \Pr[e_{ij}, j \stackrel{i}{\sim} \Sj\;|\; R_{ijs},R_{i\Sj\Ss},\tau_{i\Sj}, \rho] \; d\rho\Big)
\end{equation}

Let $I \in \{I_k\}$ be the interval containing $\theta_{i\Sj}$ and $g$ be the starting point of this interval. Note that the rectangle $R_{i\Sj\Ss}$ is associated with $I$ with probability $u = 1/2$. It is worth mentioning that, when $\theta_{i\Sj}$ is fixed, fixing $\rho$ determines $I$, and therefore, $g$ as well. We first show that $g$ is considerably smaller than $\theta_{i\Sj}$ in expectation. Note that for any $\rho$ value, $g>0.1 \theta_{i\Sj}$.  
\begin{lemma}
\label{obs:start-interval}
For any value of $\theta_{i\Sj}$, $\E[g] \leq \frac{11}{20}\theta_{i\Sj} $.
\end{lemma}
\begin{proof}
Let $k$ be an integer such that  $10^{k} \le \theta_{i\Sj} < 10^{k+1}$. Let $\alpha$ be such that $\theta_{i\Sj}=\alpha 10^{k+1}$. Note that by definition of $\theta_{i\Sj}$ and $\alpha$, we have $\frac{1}{10} \le \alpha <1$. Since $\rho$ is sampled from $(\frac{1}{10},1)$ uniformly at random, the interval including $\theta_{i\Sj}$ must be either $I_k = (\rho 10^k , \rho 10^{k+1})$ or $I_{k+1} = (\rho 10^{k+1} , \rho 10^{k+2})$. Precisely, we have $g=\rho 10^{k+1}$ when $\rho \le \alpha$, otherwise $g = \rho 10^k$. Therefore, we derive, 
\begin{align*}
	 E[g]&=\frac{10}{9}\int_{\frac{1}{10}}^\alpha \rho \cdot 10^{k+1} \;d\rho+ \int_{\alpha}^1  \rho \cdot 10^{k}\; d\rho =10^{k+1}(\frac{\alpha^2}{2}+\frac{1}{20}) \le 10^{k+1} \frac{11}{20}\alpha	 =\frac{11}{20}\theta_{i\Sj}
	\end{align*}
The inequality follows from the fact that $\frac{\alpha^2}{2}+\frac{1}{20} \le \frac{11}{20}\alpha$ when $\frac{1}{10} \le \alpha <1$. The last equality is immediate from the definition of $\theta_{i\Sj}$.
\end{proof}
The following observation will be useful in our analysis.

\iffalse
\begin{observation}
\label{ob:conf-two-interval}
For any configuration $f$ and any interval $I \in \{I_k\}_k$,  we have 
$$\sum_{(j,s)\in f} \frac{|I \cap (\hat{s},\hat{s}+p_{ij}]|}{p_{ij}} \le 1$$
\end{observation}
\begin{proof}
Due to Lemma~\ref{lem:no-bad-interval}, no bad rectangle is fully contained in the interval $I$. Thus there are at most two bad blocks in each configuration $f$overlapping  the interval $I$. Let $(j,s)$ and $(j',s)$ be these two blocks. We have $s <\frac{1}{10}(s+p_{ij})$ and $s' <\frac{1}{10}(s'+p_{ij'})$. $(j,s)$ and $(j',s')$ are disjoint. Thus, $s+p_{ij} \le s' $. Assume $I=(\rho9^k,\rho9^{k+1}]$. We have
\begin{align*}
\frac{|I \cap (\hat{s},\hat{s}+p_{ij}]|}{p_{ij}}+\frac{|I \cap (\hat{s'},\hat{s'}+p_{ij'}]|}{p_{ij'}}&=\frac{1.4s+p_{ij}-\rho9^k}{p_{ij}}+\frac{\rho9^{k+1}-(1.4s'+p_{ij'})}{p_{ij'}}
\end{align*}
\end{proof}
\fi

\begin{observation}
\label{ob:notInterior}
For any configuration $f$ and value $\theta' \le \theta_{i\Sj}$, such that no rectangle $R_{ij's'}$ in $f$ includes $\theta'$, i.e. $\theta' \notin (s', s' + p_{ij'})$ we have $\hat{L}_f(\theta_{i\Sj}) \le \theta_{i\Sj} - \min(0.34 \theta', \theta_{i\Sj}-\theta')$. 
\end{observation}
\begin{proof}
Since $\theta'$ is not contained in the interior of any rectangle in $f$, a rectangle in $f$ is either to the left or to the right of $\theta'$. The rectangles in $f$ after $\theta'$ will be shifted to the right side by at least $0.34 \theta'$. When $0.34\theta'>\theta_{i\Sj}-\theta'$, all rectangles after $\theta'$ in $f$ will be shifted to the right side of $\theta_{i\Sj}$. Thus $\hat{L}_f(\theta_{i\Sj}) \le \theta'$. When $0.34\theta'<\theta_{i\Sj}-\theta'$, the rectangles in $f$ from $\theta_{i\Sj}-0.34\theta'$ to $\theta_{i\Sj}$ will be shifted to the right side of $\theta_{i\Sj}$. Thus $\hat{L}_f(\theta_{i\Sj}) \le \theta_{i\Sj} -0.34\theta'$. Combining two cases gives the observation. 
\end{proof}

Consider the the interval $H=(0.1\theta_{i\Sj},0.97\theta_{i\Sj}]$. We upper bound $D_f$ by considering three cases.\\

\smallskip
\noindent
{\bf{Case 1: }} The interval $H$ is not a sub-interval of any rectangle (more precisely, the interval defined by the rectangle's starting and ending times) in $f$. In this case there is $\theta' \in H$ that is not interior of any rectangle in $f$. By observation~\ref{ob:notInterior}, $\hat{L}_f(\theta_{i\Sj}) \le \theta_{i\Sj}-\min\{0.03\theta_{i\Sj},0.34\theta'\} \le \theta_{i\Sj}-\min\{0.03\theta_{i\Sj},0.34\times 0.1 \theta_{i\Sj}\} \le 0.97\theta_{i\Sj}$. Thus $D_f \le z_{if}0.97 \theta_{i\Sj}$
\\
\\
\smallskip
\noindent
{\bf{Case 2: }} The interval $H$ is a sub-interval of a good rectangle $R_{ijs}\in f$. We first observe that the rectangle $R_{ijs}$ is shifted to the right side by at least $0.03\theta_{i\Sj}$. If $x_{ij}\ge \frac{9}{100}$, the rectangle $R_{ijs}$ is shifted to the right side by at least $0.34(s+0.09p_{ij}) \ge 0.34(0.09(s+p_{ij})) > 0.03\theta_{i\Sj}$. When $x_{ij}<\frac{9}{100}$, $R_{ijs}$ is good because $s > \frac{1}{10}p_{ij}$. In this case $R_{ijs}$ is shifted to the right side by $0.34s \ge 0.34(0.09s+0.91 s) \ge 0.34\times 0.09(s+p_{ij}) >0.03 \theta_{i\Sj}$. 
As $s+ 0.03\theta_{ij^*} < \theta_{ij^*}$, this means the total length of rectangles in $f$ up to time $\theta_{ij^*}$ decreases by at least $0.03\theta_{ij^*}$ due to the shifting. Thus, we have  $D_f \le z_{if}0.97\theta_{i\Sj}$.\\

\smallskip
\noindent
{\bf{Case 3: }} The interval $H$ is a sub-interval of a bad rectangle $R_{ijs}\in f$. This is the case where we utilize strong negative correlations. % We have
\begin{align*}
&\Pr[e_{ij}, j \stackrel{i}{\sim} \Sj\;|\; R_{ijs},R_{i\Sj\Ss},\tau_{i\Sj}, \rho]\\
=&\Pr[ e_{ij},\theta_{ij}\in I, j\leadsto_i I, \Sj\leadsto_i I, \sum_{j'\leadsto_i I}x_{ij'}\le 1 \mbox{ for $I$  s.t. } \theta_{ij^*} \in I \;|\; R_{ijs},R_{i\Sj\Ss},\tau_{i\Sj}, \rho]   \\
 =& \Pr[g<\theta_{ij}<\min\{\theta_{i\Sj},\hat{s}+p_{ij}\} \;|\; R_{ijs},R_{i\Sj\Ss},\tau_{i\Sj}, \rho ] \\
%% 1st dot 
&\cdot \Pr[j \leadsto_i I\;|\; e_{ij}, \theta_{ij}\in I,  R_{ijs},R_{i\Sj\Ss},\tau_{i\Sj}, \rho] \\
%% 2nd dot
&\cdot \Pr[\Sj \leadsto_i I\;|\; j \leadsto_i I, e_{ij}, \theta_{ij}\in I,  R_{ijs},R_{i\Sj\Ss},\tau_{i\Sj}, \rho]
 \;\;\; \\
%% 3rd dot
&\cdot \Pr[\sum_{j' \leadsto_i I}x_{ij'}\le 1\;|\; \Sj \leadsto_i I, j \leadsto_i I, e_{ij}, \theta_{ij}\in I,  R_{ijs},R_{i\Sj\Ss},\tau_{i\Sj}, \rho]  \\
\ge &\frac{\max\{\min\{\theta_{i\Sj},\hat{s}+p_{ij}\}-g,0\}}{p_{ij}} \cdot u \cdot u \cdot \Pr[\sum_{j' \neq j,\Sj\leadsto_i I}x_{ij'}\le 0.82] 
\end{align*}

Note that in the above equations, $I$ is a grid interval including $\theta_{ij^*}$.
Thus, fixing $R_{ij^*s^*}$ and $\tau_{ij^*} $ means fixing $I$ and its starting point $g$. In the third equation, the first probability is for the event $e_{ij}$ and $\theta_{ij} \in I$. The second and third probabilities are for the events that $j$ and $j^*$ are associated with $I$ respectively, conditioned on their $\theta$ values being in $I$ -- they are both $u = 1/2$. The last probability is lower bounded by the probability that the total height of jobs, other than $j, j^*$, assigned to $I$ is at most 0.82, as only bad jobs are associated with $I$, meaning $x_{ij},x_{ij^*} \leq 9/100$.

To keep the flow of the analysis, we defer the proof of the following lemma to the subsequent section.
\begin{lemma}
    \label{lem:pair-negative-bound}    
$\Pr[\sum_{j' \neq j,\Sj\leadsto_i I}x_{ij'}\le 0.82] \geq 0.5317$.
\end{lemma}

Assuming Lemma~\ref{lem:pair-negative-bound} holds true, we can complete our analysis. 
We have 
$$\Pr[e_{ij}, j \stackrel{i}{\sim} \Sj\;|\; R_{ijs},R_{i\Sj\Ss},\tau_{i\Sj}, \rho]> 0.1323\cdot \frac{\max\{\min\{\theta_{i\Sj},\hat{s}+p_{ij}\}-g,0\}}{p_{ij}}$$
Applying the above inequality to Eqn. (\ref{eq:conf-cont-depend-rho}) we have. 
\begin{align*}
D_f &\le z_{if}\Big( \hat{L}_{f}(\theta_{i\Sj})- \sum_{R_{i\Sj s}\in f}\hat{L}_{i\Sj s}(\theta_{i\Sj})-(1-\eta) \int_{\frac{1}{10}}^1 \frac{10}{9} \cdot 0.1323\cdot (\max\{\min\{\theta_{i\Sj},\hat{s}+p_{ij}\}-g,0\}) \; d\rho\Big)\\
&\le z_{if}\Big( \hat{L}_{f}(\theta_{i\Sj})-\sum_{R_{i\Sj s}\in f}\hat{L}_{i\Sj s}(\theta_{i\Sj})-(1-\eta) \int_{\frac{1}{10}}^1 \frac{10}{9} \cdot 0.1323\cdot (\min\{\theta_{i\Sj},\hat{s}+p_{ij}\}-g) \; d\rho\Big)\\
&\le z_{if}\Big( \hat{L}_{f}(\theta_{i\Sj})-0.1323\cdot(1-\eta)\cdot\big(\min\{\theta_{i\Sj},\hat{s}+p_{ij}\}-E[g]\big) \Big)\\
&\le z_{if}\Big( \hat{L}_{f}(\theta_{i\Sj})-0.1323\cdot(1-\eta)\cdot\big(\min\{\theta_{i\Sj},s+p_{ij}\}-E[g]\big) \Big)
\end{align*}
The last inequalities come from the fact $\min\{\theta_{i\Sj},\hat{s}+p_{ij}\} \ge \min\{\theta_{i\Sj},s+p_{ij}\}$ and Lemma~\ref{obs:start-interval}. We observe that $\hat{L}_{f}(\theta_{i\Sj})$  is also upper bounded by $ \min\{\theta_{i\Sj},s+p_{ij}\}$. To see this, assume $s+p_{ij}<\theta_{i\Sj}$ since otherwise this claim is immediate from the definition of $\hat{L}_{f}(\theta_{i\Sj})$. Note that all rectangles in $f$ starting at time $s+p_{ij}$ or later are shifted to the right by 
at least $0.34(s+p_{ij}) \geq 0.34 * 0.97 \theta_{ij^*}$. Since $s+p_{ij} \geq 0.97 \theta_{ij^*}$, this means
all those rectangles are shifted to the right side of $\theta_{i\Sj}$. Thus, we have the claim. Thus, we have, 
\begin{align*}
D_f &\le z_{if}\Big( (1-0.1323\cdot(1-\eta))\cdot\min\{\theta_{i\Sj},s+p_{ij}\}+0.1323\cdot(1-\eta) \cdot \frac{11}{20} \cdot \theta_{i\Sj}\Big)\\
&\le z_{if}\Big( (1-0.1323\cdot(1-\eta)+0.1323\cdot(1-\eta)\cdot\frac{11}{20})\theta_{i\Sj}\Big)
\end{align*}
Having $\eta < 0.589$, we can say $D_f \le z_{if} \cdot 0.976~\theta_{i\Sj} $.

%% file: 4.analysis-sched.tex
\subsubsection{Proof of Lemma~\ref{lem:pair-negative-bound}}

It now remains to prove Lemma~\ref{lem:pair-negative-bound}.

\begin{lemma}
\label{lem:conf-two-interval}
For any configuration $f$ on machine $i$ and any grid interval $I$,  we have 
$$\sum_{R_{ijs} \in f: R_{ijs} \text { is bad on $i$}} \frac{|I \cap (\hat{s},\hat{s}+p_{ij}]|}{p_{ij}} \le 1$$
\end{lemma}
\begin{proof}
Fix a machine $i$.
We first show that at most two bad rectangles from the same configuration $f$, after shifting, can overlap with a fixed grid interval $I$ on a fixed machine $i$. To streamline our analysis, by scaling all time points uniformly, we assume wlog that $I = (1/10, 1)$. For the sake of contradiction, say there are more than two bad rectangles from $f$ overlapping with $I$ after shifting. Suppose $\hat{R}_{ij's'}$ is the last bad rectangle overlapping with $I$ and $\hat{R}_{ijs}$ is the second to the last bad rectangle overlapping with $I$. Here we use $\hat{R}$ to denote the rectangle $R$ after shifting. 

The  proof idea is to show that even if $\hat{R}_{ij's'}$ barely overlaps with $I$, the second to the last rectangle $\hat{R}_{ijs}$ must start before $I$. So, we can assume $\hat{s}' = 1$.
For notational convenience, let $p := p_{ij}$ and $p' := p_{ij'}$. 
As $R_{ij's'}$ is bad, we know $\hat{s}' = 1.34s'$.
From the observation that the two rectangles do not overlap and their relative order doesn't change by the shifting, we have $s' \geq s+ p$. Since ${R}_{ijs}$ is bad, we have $s \leq (1/10)p$. Therefore, we have $1 = 1.34 s' \geq 1.34 (s+p) \geq 1.34 (11s) = 14.75s$. Thus, we have $\hat{s} = 1.34s \leq 1/11$. This means $\hat{R}_{ijs}$ starts before $I$ does. This, $\hat{R}_{ijs}$ is the first bad rectangle from $f$ overlapping with $I$, a contradiction. 

From the above argument, it is straightforward to see that if the first bad rectangle overlapping with $I$ is fully contained in $I$, then it is the only bad rectangle from $f$ overlapping with $I$. In this case, the summation is exactly one. To see this is in fact the maximum of the summation, fix two adjacent bad rectangles $\hat R_{ijs}$ and $\hat R_{ij's'}$ from the same configuration, move $I$ to the left or to the right. Since $p' > p$ from the above argument, it is easy to see that when we increase $\hat R_{ij's'}$'s overlap with $I$ by one unit and decrease $\hat R_{ijs}$'s overlap with $I$ by one unit, the summation decreases. This implies the summation is at most one, as desired.  
\end{proof}

\begin{lemma}
\label{lem:expect-height-interval}
    For any machine, the expected total height of jobs associated with a grid interval $I$ on the machine is at most $u$.
\end{lemma}
\begin{proof}
We have
\begin{align*}
\E[\sum_{j \leadsto_i I} x_{ij}]&=\sum_{j} x_{ij}\Pr[j \leadsto_i I]\\
&=\sum_j x_{ij}\sum_{s: R_{ijs} \text { is bad on $i$}} \frac{x_{ijs}}{x_{ij}} \Pr[j \leadsto_i I\;|\; R_{ijs}]\\
&=\sum_j x_{ij}\sum_{s : R_{ijs} \text { is bad on $i$}}\frac{x_{ijs}}{x_{ij}} \cdot u \cdot \frac{|I \cap (\hat{s},\hat{s}+p_{ij}]|}{p_{ij}}\\
&=u \sum_f \sum_{R_{ijs} \in f:  R_{ijs} \text { is bad on $i$}} z_{if}  \cdot \frac{|I \cap (\hat{s},\hat{s}+p_{ij}]|}{p_{ij}}\\
&\leq u \sum_f z_{if} \quad \quad \quad \text{(Lemma~\ref{lem:conf-two-interval})}
\\
& \le u
\end{align*}
The third equality follows since $j$ is associated with $I$ when $\theta_{ij}\in I$, which occurs with probability $\frac{|I \cap (\hat{s},\hat{s}+p_{ij}]|}{p_{ij}}$ conditioned on $R_{ij} = R_{ijs}$, and an independent fair coin $(u = 1/2)$ gives a head.
\end{proof}

\iffalse
Using Markov's inequality and Lemma~\ref{lem:expect-height-interval}, we have
\begin{equation}
\label{eq:height-prob}
\Pr[\sum_{j' \neq j,\Sj\leadsto_i I}x_{ij'}\le 0.8]  \ge 1- \frac{u}{0.8}=3/8
\end{equation}
\fi

To prove Lemma~\ref{lem:pair-negative-bound}, we use the following well-known concentration inequality.
\begin{theorem}[Theorem 2.3 of \cite{mcdiarmid1998concentration}]
    \label{thm:cher_b}
Let $Z$ be the sum of $n$ independent random variables where each random variable takes value in $[0, K]$. Let $\mu = \E[Z]$. Then for any $\lambda \in [0, 1]$, we have
\begin{equation*}
\Pr\Big[Z\geq (1 +
\lambda) \mu \Big] \leq e^{-\lambda^2\mu/3K}.
\end{equation*}
\end{theorem}

Let $Z_{j'}$ denote the indicator variable for the event $j' \leadsto_i I$. Then, we know $\E [Z := \sum_{j' \neq j, j^*} x_{ij'} Z_{j'}] \leq u = 1/2$. Note that $x_{ij'} Z_{j'} \leq 9/100$, as $j'$ is associated with interval $I$ only when it is bad on machine $i$; thus, we can set $K = 9/100$. Further, $\{Z_{j'}\}_{j' \neq j, j^*}$ are independent from one another. Therefore, by setting $\lambda = 0.64$, we obtain,
\begin{align*}
&\Pr[\sum_{j' \neq j,\Sj\leadsto_i I}x_{ij'}\le 0.82] 
= 1 - \Pr[\sum_{j' \neq j,\Sj\leadsto_i I}x_{ij'} \geq 0.82] \\
&\geq  1 -  \exp( -0.64^2 \cdot 0.5 / (3 * 0.09)) > 0.5317.
\end{align*}
This completes the proof of  Lemma~\ref{lem:pair-negative-bound}.